\documentclass[11pt,letterpaper]{article}

\title{Communication complexity of Nash equilibrium in potential games}
\author{Yakov Babichenko\footnote{Technion, Israel institute of Technology. E-mail: yakovbab@technion.ac.il.} \ and Aviad Rubinstein\footnote{Stanford University. E-mail: aviad@cs.stanford.edu.}}
\date{}

\usepackage{complexity}
\newclass{\CLS}{CLS}
\usepackage{amssymb}
\usepackage{amsmath}
\usepackage{amsthm}
\usepackage{float}
\usepackage{hyperref}
\usepackage{tikz}
\usetikzlibrary{calc,patterns,shapes}
\usetikzlibrary{arrows.meta}

\usepackage[left=1.00in, right=1.00in, top=1.00in, bottom=1.00in]{geometry}

\newcommand{\ra}{\rightarrow}
\newcommand{\N}{\mathbb{N}}
\newcommand{\Real}{\mathbb{R}}
\newcommand{\Ex}{\mathbb{E}}
\newcommand{\Z}{\mathbb{Z}}
\newcommand{\Prob}{\mathbb{P}}
\newcommand{\calL}{\mathcal{L}}
\newcommand{\oa}{\hat{a}}
\newcommand{\oC}{{C^{\conv}}}

\newcommand{\ob}{\overline{b}}
\newcommand{\ophi}{\overline{\phi}}
\newcommand{\oal}{\overline{\alpha}}
\newcommand{\obe}{\overline{\beta}}

\newcommand{\0}{\textbf{0}}

\usepackage{bbm}
\newcommand{\1}{\mathbbm{1}}

\newcommand{\gsm}{28}
\newcommand{\gs}{29}
\newcommand{\gsp}{30}
\newcommand{\tgs}{58}
\newcommand{\hgs}{15}
\newcommand{\thgs}{88}
\newcommand{\eol}{\textsc{EndOfLine}}

\newcommand{\Pyr}{\textsl{Pyr}}

\newcommand{\uv}{\underline{v}}
\newcommand{\ov}{\overline{v}}

\DeclareMathOperator{\conv}{conv}

\DeclareMathOperator{\supp}{supp}

\newtheorem{lemma}{Lemma}
\newtheorem{theorem}{Theorem}
\newtheorem{oq}{Open Question}
\newtheorem{corollary}{Corollary}
\newtheorem*{corollary*}{Corollary}
\newtheorem*{proposition*}{Key Proposition}
\newtheorem{proposition}{Proposition}
\newtheorem*{mtheo*}{Main Theorem}
\theoremstyle{definition}
\newtheorem{definition}{Definition}
\newtheorem{fact}{Fact}
\theoremstyle{remark}
\newtheorem{remark}{Remark}
\newtheorem*{remark*}{Remark}

\begin{document}

\maketitle
\begin{abstract}
We prove communication complexity lower bounds for (possibly mixed) Nash equilibrium in potential games. In particular, we show that finding a Nash equilibrium requires  $\poly(N)$ communication in two-player $N \times N$ potential games, and $2^{\poly(n)}$ communication in $n$-player two-action games.
To the best of our knowledge, these are the first results to demonstrate hardness in any model of (possibly mixed) Nash equilibrium in potential games. 
\end{abstract}

\setcounter{page}{0}
\thispagestyle{empty}
\newpage

\section{Introduction}

Potential games~\cite{MS} is a fundamental class of games that captures a variety of scenarios from routing and congestion games to Cournot oligopolies. 
What all these games have in common is that they admit pure equilibria, and are even equipped with a simple and natural algorithm for finding them: the best-response dynamics, where in each step a single player deviates to her best-response action%
\footnote{The same is true for generalizations such as {\em ordinal} and {\em weighted} potential games~\cite{MS}, but our hardness results hold even for the most restrictive definition of exact potential games.}. 

Unfortunately, although the best-response dynamics algorithm for finding pure equilibria in finite potential games is guaranteed to converge, it can take exponential time~\cite{FPT} and in general finding any pure equilibrium requires exponential communication~\cite{CCPLS} and computational variants are known to be \PLS-complete~\cite{PLS,SY91,FPT,S2,AS08,ARV08,HHKS13}.

In this paper we study the communication complexity of computing a {\em mixed} Nash equilibrium in potential games, arguably the last remaining hope for efficient equilibrium computation in this class. 
Backing this algorithmic hope is the observation that none of the aforementioned hardness results generalize to mixed equilibria; and none of the known hard instances for mixed equilibria in general games admit potential functions.
Daskalakis and Papadimitriou~\cite{CLS} formalize these intuitive barriers to proving hardness for computational variants of this problem by showing that they lie in $\PPAD \cap \PLS$, suggesting that they may be strictly easier than both pure equilibrium in potential games and mixed equilibrium in general games. 
We overcome these technical barriers, and prove the first (in any model)  non-trivial intractability results for Nash equilibrium in potential games.

Our main result is a lower bound in the following simple and natural {\em communication complexity} formulation of equilibrium computation:
each party knows the utility function of one player, and their goal is to jointly compute an equilibrium of the game%
\footnote{Note that since potential games admit succinct equilibria, once any party learns an equilibrium she can can broadcast it with negligible communication to all other parties.}.
Lower bounds in this model hold without imposing restrictions on the computational power or strategic behavior of the parties.

\begin{mtheo*}[Informal, see Theorems~\ref{theo:2p} and~\ref{theo:np}] \hfill
\begin{description}
\item[Two-player games] 
The randomized communication complexity of computing a Nash equilibrium in two-player $N\times N$ potential games is at least $\Omega(N^c)$ for a constant $c>0$.
\item[Binary-action games]
The randomized communication complexity of computing a Nash equilibrium in $n$-player $2$-action potential games is at least $2^{\Omega(\sqrt{n})}$.
\end{description}
\end{mtheo*}

Note that these results imply the same bounds in the query complexity model. We emphasize that even in the simpler-to-analyze query model, no hardness result on mixed Nash equilibria of potential games was known.

\begin{remark}[Approximate Nash equilibrium]
Our results extend to the hardness of approximate Nash equilibrium with polynomial (respectively exponential) approximation error for two-player (respectively $n$-player) games. 
Note that this is roughly the strongest hardness of approximation we can hope to show since the best-response dynamics give an $\epsilon$-approximate equilibrium within $O(n/\epsilon)$ steps. 
\end{remark}

\subsection{Communication complexity and game theory}

Communication complexity is a particularly attractive measure in game theoretic applications because there is a natural correspondence between protocol parties and game players, and it evades questions of how agents represent and access their utility functions. 
This connection has been studied extensively in the context of Combinatorial Auctions~\cite{NS06,BNS07,Feige09,DV13,DobzinskiNO14,Dobzinski16b,Assadi17,BravermanMW17,EzraFNTW19} and also Price of Anarchy~\cite{Roughgarden14}, Fair Division~\cite{BranzeiN19,PlautR19} and equilibrium computation~\cite{CS,HMan,CCNash,GR18,GS18,CCPLS}.

The communication complexity model is of particular interest in the context of equilibria computation. As was shown by \cite{CS} (see also \cite{HMan}), the communication complexity of a solution concept captures (up to a logarithmic factor) the rate of convergence of natural dynamics (\emph{uncoupled dynamics} \cite{HMas,HMan}) to this solution concept. Thus, a lower bound on the communication complexity implies that there exists no natural dynamics that lead to this solution concept in a reasonable time. In particular, the result of \cite{CCPLS}, implies that no dynamics can lead to a pure Nash equilibrium in potential games in a reasonable time. We strengthen this negative result and prove that the same is true for the wider solution concept of mixed Nash equilibria. We note that the performance of specific dynamics such as the better-reply dynamics in potential games has been previously studied \cite{FPT,S2}. Communication complexity results, on the other hand, imply a slow rate of convergence for \emph{all} dynamics.

\subsection{Complexity theory context}

The complexity of Nash equilibria in general games has been extensively studied in the past two decades for different complexity models, including computational complexity \cite{LMM,SS04,DGP,CD,EY10,Mehta14,Rub2p,KM18}, query complexity \cite{QCNash,GR,FS16,CCT15,Rub2p}, and communication complexity \cite{HMan,CCNash,GR18}. This extensive study indicates that in general Nash equilibrium is a hard task in all mentioned above models.

The computational hardness of Nash equilibrium in the general case raises the question of whether Nash equilibrium can be computed efficiently for \emph{classes of games} with economic significance. This question has been extensively studied as well, in particular for the classes of graphical games \cite{DGP,RubGraph,DFS20}, anonymous games \cite{DPAnon,CDO}, congestion games \cite{FPT,HSko,CSin, CLS}, and the closely related class of potential games \cite{FPT, CLS}.

In multi-player potential games, it is shown in \cite{FPT}, that computation of an \emph{exact pure} Nash equilibrium cannot be done efficiently, unless $\PLS=\P$. The hardness of pure Nash equilibrium result holds also in the query complexity model \cite{NBlog}. 
Recently, \cite{CCPLS} showed the hardness of pure Nash equilibrium also in the \emph{communication complexity} model.

All the above-mentioned hardness results (in all three complexity models) for pure Nash equilibrium in potential games, do not contradict the hypothesis that maybe Nash equilibrium, not necessarily pure, can be computed efficiently (in the computational model) or can be learned quickly by some dynamics (in the communication model). In fact, to the best of our knowledge, no hardness result is known for Nash equilibrium in potential games in any setting. Moreover, the existing techniques for the hardness of pure equilibrium are not helpful for the mixed Nash equilibrium problem: existing reductions for pure equilibrium focus on discrete objects (typically graphs), and the resulting games allow spurious \emph{mixed} Nash equilibria.

The complexity of Nash equilibrium in potential games is interesting from two aspects. First, as mentioned above, from the game theoretic perspective it is natural to ask how fast can players learn/compute any Nash equilibrium (not necessarily pure). 

From the theoretical computer science perspective, the existence of Nash equilibrium in potential games has two completely different non-algorithmic proofs: Nash's theorem, which says that any finite game has a Nash equilibrium, relies on Brouwer's fixed point theorem, which in turn is based on a parity argument; this proof is captured by the complexity class \PPAD~\cite{PPAD}.
For potential games, any sequence of best-reply updates is monotone increasing in the potential function, and must, therefore, converge to a local maximum, which is an equilibrium; this proof, which guarantees a pure Nash equilibrium, is captured by the complexity class\footnote{\PPAD, \PLS, and other subclasses of \TFNP~have formal communication complexity analogs (see \cite{GKRS19-CC-TFNP}). Interestingly, "Communication \PPAD" is not known to contain the problem of finding a Nash equilibrium in two-player games (even for constant approximation error). For a binary-action many-player potential games, standard techniques (e.g.~\cite{DGP}) suffice to show that computing an approximate Nash equilibrium (with exponentially-small error) is indeed in ``Communication $\PPAD\cap \PLS$''.} \PLS.
The complexity of problems whose solution admits both \PPAD~and \PLS~existence proofs remains perhaps the least understood within the study of total search problems (\TFNP). In particular, \cite{CLS} show that the intersection of these classes, $\PPAD \cap \PLS$ includes congestion games, implicit congestion games, and network coordination games, all of which are subclasses of many players potential games.

Recently, this area has been very active, with exciting progress including query complexity lower bounds, \CLS-completeness, and hardness based on cryptographic assumptions for problems defined with a circuit \cite{ConMap1, ConMap2, HY17-crypto, GDHKS18-CLS,EPRY20}. In particular, an exciting very recent breakthrough of~\cite{FGHS20} shows that in fact $\CLS = \PPAD \cap \PLS$!
However, proving completeness for any natural problem (whose definition is not through a circuit) in this class remained open.
In a follow-up paper~\cite{BR21}, inspired by the current paper, we resolve this open problem by showing, among other results, that Nash equilibrium in congestion games is $\PPAD \cap \PLS$-complete.


\begin{remark}[Promise vs total for potential games] In computational complexity there is an important distinction between {\em total problems} (where every instance has a solution), and {\em promise problems} (where every instance that satisfies a certain ``promise'' has a solution, and the algorithm is only required to succeed on those instances). 
On this issue, our hardness results enjoy the best of both worlds: (i) all our hard instances satisfy the promise, i.e. they are actual potential games; and (ii) one can define a total%
\footnote{By Nash's theorem, even for non-potential games, finding an equilibrium is obviously already total  in the \PPAD-sense.} problem of ``find an equilibrium or certify that the game is not potential''. Indeed, by~\cite{MS, CCPLS}, if the game is not a potential game, then there exists succinct certificate that the game is not potential. Moreover, this certificate can be found efficiently. (See~\cite[Section 4.1]{CCPLS} for details.)
\end{remark}

\subsection{Congestion vs Potential games}\label{sub:congestion}
Congestion games~\cite{Rosenthal73} are a formal framework for studying how selfish agents choose their routes in a congested network. 
They are extremely well-studied in Game Theory, Economics, and in particular Computer Science, where routing traffic in decentralized  computer networks (e.g.~the Internet) is an important practical question.

Congestion games are defined over a ground set of {\em facilities} (``primary factors'' in Rosenthal's original paper).
Each player's action corresponds to a subset of the facilities. The cost of each facility is a function of the number of players choosing it, and the total cost to each player is the sum of costs in her subset.

To model congestion games as a communication complexity problem, we assume that every player knows every other player's feasible subsets, as well as the cost function on every facility in the union of her feasible subsets.
The private information to players is the cost function on facilities only feasible to them.

As already mentioned, congestion games are a sub-class of potential games. Furthermore, every potential game is isomorphic to a congestion game~\cite{MS}. 
A simple variant of this reduction shows that our communication complexity lower bounds for potential games extend to congestion games.

\begin{corollary*}[Congestion games; informal- see Corollaries~\ref{cor:cong2} and~\ref{cor:congn}] \hfill
\begin{description}
\item[Two-player congestion games] 
The randomized communication complexity of computing a Nash equilibrium in a two-player congestion game with $N$ facilities is at least $\Omega(N^c)$ for a constant $c>0$.
\item[Binary-action congestion games]
The randomized communication complexity of computing a Nash equilibrium in $n$-player $2$-action congestion games is at least $2^{\Omega(\sqrt{n})}$.
\end{description}
\end{corollary*}

\subsection{Techniques}\label{sub:techniques}

At the core, much of the recent progress on the hardness of computing Nash equilibrium in general games%
\footnote{The formulation of the problem we study is a natural extension of two lines of work: complexity of (mixed) Nash equilibrium in general games; and complexity of pure equilibrium in potential games. But in terms of techniques we are much closer to the former, and in particular tools developed to deal with the continuum of candidate solutions.}
 followed a common very rough blueprint:
\begin{description}
\item[End-of-Line] Begin with the problem of finding the end of a line%
\footnote{For our construction, we essentially use the End-of-{\em Metered}-Line~\cite{HY17-crypto, ConMap2} variant of the problem, which also encodes a potential function on the line.} in a graph. In computational complexity, it is \PPAD-hard by definition. In query complexity, it is typically not difficult to prove unconditional hardness, which can be lifted to communication complexity using ``simulation theorems''.
\item[Brouwer's fixed point] Reduce End-of-Line to finding a fixed point of a continuous, Lipschitz function $f: D \rightarrow D$ (where $D$ is some convex domain $D$, typically $[0,1]^n$).
This is typically the most technical part, as it reduces a discrete problem to a continuous one - this is crucial for {\em mixed} equilibria.
\item[Imitation game]
 Alice and Bob pick $a,b$ from the above convex domain $D$. Alice's objective is to minimize $||a-b||$, while Bob's objective is to minimize $||a-f(b)||$.
\end{description}

To best understand our paper, let us start from the end:
The nice and critical feature of the imitation game construction is that Alice's unique best reply to any mixed strategy of Bob is simply the center of mass%
\footnote{The center of mass simplification only holds when we take $||\cdot||$ to be the semimetric $||\cdot||_2^2$. In our construction, we need semimetrics that increase much faster, e.g.~$||\cdot||_2^n$. This is crucial in order to compete with exponentially large jumps in potential. The center of mass intuition still holds, but the details are more subtle.} of $b$. Thereafter we can restrict the analysis to Alice's pure actions, which are much more tractable.
For our purposes, the main caveat is that this imitation game is not a potential game!

A key novel idea that we introduce is a game where the above imitation gadget is additively separable from the complex part of the utilities which actually encodes the reduction.
Specifically, we construct the following two-player potential game.
Alice and Bob respectively choose points $a,b$ from some convex domain $D$. The utilities are given by 
\begin{align}\label{eq:imitation}
\begin{split}
u_A(a,b)&=-||a-b||-\epsilon \phi(a), \\
u_B(a,b)&=-||a-b||-\epsilon \phi(b),
\end{split}
\end{align}
where $\phi$ is a hard-to-locally-optimize potential function.
Namely, players are primarily incentivized to play close to each other's strategy; in addition, each player has a mild incentive to play a point with high potential value. Intuitively, this combination of incentives creates a situation where at any point which is not a local minimum of the potential, a player prefers to deviate slightly to a close-by point with higher potential; however she does not want to deviate too far, because then the distance penalty becomes more significant than the gain in the potential. Hence, we expect the equilibria to be the local minima of $\phi$. The challenging part is to apply this intuition to a \emph{mixed} Nash equilibrium analysis. 

\subsubsection*{The potential function}

Perhaps the most technically elaborate part in our proof is the construction of a suitable potential function. Our goal is to embed (a variant of) the discrete End of Line problem as a continuous potential function $\phi: D \rightarrow \Real$.  In particular we require that all local maxima of $\phi$ correspond to the end of the line%
\footnote{We actually require our potential function to satisfy a stronger and somewhat subtle desideratum --- see Section~\ref{sec:domdir} for details.}. 

In our experience with complexity of non-potential games, a very useful construction is that of~\cite{HPV89} 
which embeds a line into a continuous, Lipschitz function $f: D \rightarrow D$ whose fixed points are located near the end-of-line point. 
At a very high level, it is helpful to think of the corresponding displacement function, $g(x) := f(x)-x$, as the gradient of our potential function $\phi$. Then, any local maxima 
of $\phi$ would indeed correspond to a fixed point of $f$.
$$ \text{Wishful thinking:\;\;\;} \nabla \phi = g.$$
This analogy, while inspiring, is over-optimistic; for example, there is no reason why $g$ should satisfy the Gradient Theorem (hence it is not a valid gradient of any function).
One of our main technical contributions is in constructing and analyzing an explicit potential function $\phi$ that embeds the End-of-Line problem a-la the $n$-dimensional variant of~\cite{HPV89} construction. 
See Section \ref{sec:hpv} for details.
We remark that related constructions were known for two dimensions~\cite{Vavasis93,HY17-crypto,BM20,CDHS20}.

The combination of the above ideas brings us most of the way to our result for two-player games in the query complexity model.

\subsubsection*{Binary-action games}
To extend our result to binary-action $n$-player games, we distribute the task of reporting each coordinate of the vector $a$ (respectively $b$) to a subset of players.
Even though each player only reports $0$ or $1$, taking the average over those reports allows us to construct a richer game, with strategy profiles that can (approximately) represent any point in the convex domain of $\phi$. 

Another advantage that we introduce into our reduction by averaging over many ($\Theta(\sqrt{n})$ players) is that even if players mix their individual strategies, the average concentrates around its expectation with exponentially high probability. This ensures that we can again treat these vectors almost as pure strategies\footnote{Similar idea of replicating players has been introduced by \cite{CCT15} in a different context of reducing approximate well-supported Nash equilibrium to approximate Nash equilibrium.}.

\subsubsection*{Communication complexity}
An important issue that we glossed over so far is that the game~\eqref{eq:imitation} is actually trivial to solve in communication complexity: since the potential function is explicitly encoded in the players' utility function, they can find its maximum without any communication! 

Similarly to \cite{CCNash,GR18,CCPLS}, in order to obtain a hard game, we partition, for each $x \in D$, the {\em local information} about $\phi(x)$ between Alice and Bob as prescribed by a ``simulation theorem'' that ``lifts'' the query complexity of End-of-Line to communication complexity (see Section \ref{sec:cc}). In order to compute the potential at $a=b$, they need to combine the respective pair of local information. 

We introduce additional components to their strategy space which correspond to reporting their local information about $\phi(a)$ (respectively $\phi(b)$). Their utilities have additional components that incentivize them to truthfully report this local information. It is crucial that the incentive to truthfully report is much greater than what they can gain from an increase in the potential function due to a non-truthful report.
The updated two-player game looks roughly like this:
\begin{align}\label{eq:informal-game-Alice}
\begin{split}
u_A(a,b)&=-||a-b||+\epsilon\1\{\text{local info true}\} - \epsilon^2 \phi(a)  \\
u_B(a,b)&=-||a-b||+\epsilon\1\{\text{local info true}\}-\epsilon^2 \phi(b).
\end{split} 
\end{align}
For binary-action games, we need additional (subsets of) player to report the local information.
Here, each of Alice's players has identical interest utility a-la~\eqref{eq:informal-game-Alice} (this preserves the potential game property).

 But now we face a new subtle obstacle: in order to move in the direction of the potential function's gradient, two (subsets of) players must deviate simultaneously - both those responsible for some coordinate of $a$, and those responsible for the truthful reporting. In other words, this means that we would have spurious equilibria at potential non-maximizers just because there are no improving unilateral deviations. 
We note that this issue does not occur in general games because the coordinate players are not incentivized to help the local information players report remain truthful. A similar issue does arise in the analysis of pure Nash equilibria in potential games \cite{CCPLS}.

We resolve this issue by alternating between two subsets of local information players. At a high level, Alice is allowed to use the report of either subset. When one subset's report is updated and stable, the second subset can update their report. Once the second subset updated their report, the coordinate players can use that stable report to advance toward increasing the potential function. Then the first subset of local information players can update their reports as well. See Section \ref{sec:localinfo} for details.

\subsection*{Open problems}
The main open problem left in previous versions of this paper was to characterize the {\em computational} complexity of related problems in $\PPAD \cap \PLS$. 
Major progress has been made on this problem since, and we now know that congestion games are $\PPAD \cap \PLS$-complete~\cite{FGHS20,BR21}. 

Another, more technical, open problem is to obtain fine-grained variants of our lower bounds, namely:
\begin{oq}[Fine-grained communication complexity of potential games]
Close the gaps between our results and known upper bounds on communication complexity
\begin{description}
\item[Two-player games]
$\Omega(N^c)$ vs $\tilde{O}(N)$.\footnote{To see why $\tilde{O}(N)$ communication suffices, consider the best-response dynamics, where at each iteration a player with maximal improving deviation changes their action. Because the potential function increases at each improving-response deviation, no player can deviate to the same action more than once.} 
\item[Binary action games] $2^{\Omega(\sqrt{n})}$ vs $2^{O(n)}$.
\end{description}\end{oq}
For two-player games, we remark that an analogous gap was first left open by our paper on communication complexity of approximate Nash equilibrium in general games~\cite{CCNash}, and later completely (up to lower order terms in the exponent) closed in followup work~\cite{GR18}. The latter tight bound relies on mixed strategies to optimally encode vertices in the End-of-Line graph. However, potential games always admit pure equilibria, so closing this gap will likely require new ideas.

For $n$-player binary-action games, our sub-exponential bound of $2^{\sqrt{n}}$ is a result of our replication technique: We have a potential function that takes as input $n$ real variables in $[0,1]$, and we want to represent each variable using a team of $m$ $\{0,1\}$-action players; specifically, we use the average of the actions on the $i$-th team as the $i$-th input to the potential function. When players use mixed strategies, their empirical average can deviate from the intended expectation. Since our potential function is exponentially sensitive, we need exponential concentration in every coordinate to guarantee correct behavior. This forces us to set $m=\Theta(n)$, i.e.~we have $\Theta(n)$ players representing each coordinate, which is a quadratically inefficient representation. See Section~\ref{sec:rep} for details.

There may be more efficient representations for $n$-player games. For example, our result for two-player games uses a different technique that we term ``high degree imitation''. En route, in Section~\ref{sec:hd}, we also prove the following exponential {\em query} complexity lower bound for $n$-player constant-action games:
\begin{corollary*}[Corollary~\ref{cor:const}]
The query complexity of finding a Nash equilibrium (possibly mixed) in $n$-player $30$-action potential games is at least $2^{\Omega(n)}$.
\end{corollary*}

\subsubsection*{Roadmap}
In Section \ref{sec:pre} we present the notations and some preliminaries on potential games and on the end-of-line problem. Section \ref{sec:hpv} presents the main ingredient in the proofs: an embedding of a line into a potential function. Section \ref{sec:qc} shows how we can utilize this embedding in order to prove query lower bounds on (mixed) Nash equilibrium in potential games. Finally, Section \ref{sec:cc} shows how we can lift these query results to communication complexity. Section~\ref{sec:congestion} formalizes the extension to congestion games.

\section{Preliminaries}\label{sec:pre}

\subsection{Notations}
We denote $[n]:=\{1,2,...,n\}$ and $[n]_0:=\{0,1,...,n\}$. 
We denote by $e^n_1,...,e^n_n$ the unit vectors of $\Real ^n$. In cases when it is obvious what is the dimension, we will simply write $e_1,...,e_n$.
We denote by $d_1(x,y)=\sum_i |x_i-y_i|$ the $\calL_1$ norm. 
We denote by $d_\infty (x,y)=\max_i |x_i-y_i|$ the $\calL_\infty$ norm. We denote by $B(x,r):=B_\infty(x,r)=\{y:d_\infty(x,y)\leq r\}$ the ball of radius $r$ around a point $x$ in the $\calL_\infty$ norm, and similarly for a set $X\subset \Z^n$, we denote $B(X,r):=B_\infty(X,r)=\{y:\exists x\in X \text{ s.t. } d_\infty(x,y)\leq r\}$.

\subsection{Potential games}

\begin{definition}[Potential games]\hfill

An $n$-player game with action space $X$ is a {\em potential game} if there exists a {\em potential function} $\phi:X \rightarrow \Real$ such that for every player $i$, actions $x_i,y_i$ for $i$ and actions $z_{-i}$ for all other players,
\begin{gather*}
u_i(x_i,z_{-i}) - u_i(y_i,z_{-i}) = \phi(x_i,z_{-i}) - \phi(y_i,z_{-i}).
\end{gather*}
In words, unilateral change of an action affects the utility of the deviator precisely by the change in the potential. 
\end{definition}

\begin{fact}\label{rem:pot}
The following are potential games:
\begin{itemize}
\item \emph{Identical interest games}; i.e., games where $u_i=U$ for every player $i$ (and some function $U$).
\item \emph{Opponent-independent games}; 
i.e., games where $u_i(x_i,y_{-i})=u_i(x_i,z_{-i})$ for every player $i$, action $x_i$ for this player, and actions $y_{-i},z_{-i}$ for the rest of the players. 
\item A generalization of the previous two that we call {\em team-opponent-independent game}, where the players are partitioned into subsets (teams) and the utilities of players within each team are identical, and independent of players outside the team ($u_i(x_S,x_{-S}) = U_S(x_S)$).
\item Sum of potential games; I.e., if $H$ and $H'$ are two potential games, so is the game $G=H+H'$. 
\end{itemize}
\end{fact}
\begin{proof}
We can see that a sum of potential games is a potential game by taking the sum of the respective potential functions.

Identical-interest games are potential games by taking the potential function to be the identical interest; i.e., $\phi := u_1$.

The above argument extends to identical interest games that are augmented with irrelevant players (who receive zero utility and do not affect other players' utilities). Finally, notice that team-opponent-independent games can be written as sums of such augmented identical interest games (one for each team).
\end{proof}

\subsection{End-of-Line Problem}

Our starting point is a variant of the End-of-Line problem over the Pyramid graph, which is known to be hard in both the query model and the communication model. We embed this problem to an end-of-line over a hypercube. This embedding induces an end-of-line problem over the hypercube with some additional structure of the line.

The Pyramid graph $\Pyr(T)$ is a directed graph whose vertices are given by $V_{P}=\{(x,y):x,y\in [T]_0 \text{ and } x+y\leq T\}$, and the edges are given by $E_{P}=\{(v,v+e_i):v\in [T]_0^2,i=1,2\}.$ Simply speaking, the Pyramid graph is the directed two-dimensional grid.

\paragraph{Query complexity} 
Let $L_{P}=(l_{P}(0),l_{P}(1),...,l_{P}(T))$ be a line of length $T$ in the Pyramid graph with starting point $\0_2$; i.e., $l_{P}(t)\in V_{Py}$, $L_{P}(0)=(0,0)$, and $(l_{P}(t),l_{P}(t+1))\in E_{P}$. Note that any line over the Pyramid graph is, in fact, a \emph{metered line} (see \cite{HY17-crypto}). Namely, if a line goes through a given vertex $(x,y)$ we know that the line went exactly $x+y$ steps so far, and exactly $T-x-y$ steps remained.

In the query problem of $\eol(\Pyr(T))$ the input is a line $L_P$. The output is the end-of-line $l_P(T)$. The queries are vertices $v=(v_1,v_2)\in V_P$, and the answer is a triple $(t_v,s_v,p_v)\in \{0,1\}^3$ that describes whether the line goes through $v$ (formally, $t_v=\1_{v\in L_P}$) and if so, also it reports whether the successor increases the first or second coordinate (formally, $s_v=\1_{v, v+e_1\in L_P}$) and whether the predecessor decreases the first or second coordinate (formally, $p_v=\1_{v, v-e_1\in L_P}$). It is well known (see e.g., \cite{Ald}) that $QC(\eol(\Pyr(T)))=\Theta(\sqrt{T})$.

\paragraph{Communication Complexity} 
We abuse notation and also use $\eol(\Pyr(T))$ to denote the following communication variant of the End-of-Line problem.
The input $(t_v,s_v,p_v)_{v\in V_P}$ is distributed between Alice and Bob as follows. For every vertex $v\in V_P$ Alice holds a triple of arrays $t^A_v,s^A_v,p^A_v\in \{0,1\}^3$. Bob holds a triple of indices $t^B_v,s^B_v,p^B_v\in [3]$. Their goal is to compute the end of the line $(t^A_v(t^B_v),s^A_v(s^B_v),p^A_v(p^B_v))_{v\in V_P}$. \cite{HN12} and \cite{GP14} have proved that this lifting of the problem to communication is as hard as the query problem, even for randomized communication model. 
\begin{theorem}[\cite{HN12,GP14}]\label{thoe:ccEOL}
$CC(\eol(\Pyr(T))=\Theta(QC(\eol(\Pyr(T)))=\Theta(\sqrt{T})$.
\end{theorem}

\section{Embedding a Line as a Potential Function}\label{sec:hpv}

We embed the End-of-Metered-Line over the pyramid graph problem into a problem of finding a local maximum of a potential function. 
We define our hard potential function on the integer points of $[\gs]_0^{n+1}$; 
for fractional points, we later use the multilinear extension on each length-$1$ subcube.

Crucially, we use a somewhat unusual notion of ``local maximum'', 
which we call ``small cube with(out) a dominating direction''. 
This notion is well defined for a function over the discrete hypergrid and also captures precisely the properties that we later use for our equilibrium analysis.
See Definition~\ref{def:dd} for details.

Our construction proceeds in two steps. The first (easy) step embeds the two-dimensional End-of-Line instance into a line in $[\gs]_0^{n+1}$. The second (involved) step defines the potential functions with respect to that line.  

\subsection{Embedding $\eol(\Pyr(T))$ into $n+1$ dimensions}

Our first step is embedding the two-dimensional End-of-Line instance in $n+1$ dimensions, where $n := 2\log_2(T)$. This step is further broken into two sub-steps.

The first sub-step encodes each coordinate of the planar instance with $n/2$ bits using Gray code%
\footnote{Gray code is a binary encoding of integers where the encodings of any two consecutive integers differs by exactly one bit~\cite{Black20}}. This ensures that if $(x,y)$ and $(x,y\pm1)$ (respectively, $(x\pm1,y)$) are neighbors in $\Pyr(T)$, the corresponding embedded vertices are also neighbors in the discrete $n$-dimensional hypercube graph. Note that such an embedding preserves the property that the line is metered. Namely, if a line goes through a vertex $v$, the index of this vertex in the line can be determined from $v$.

The second sub-step embeds the $l$ line from the $n$-dimensional hypercube graph $\{0,1\}^n$ into a new line $l'$ on the $(n+1)$-dimensional grid $G_{n+1}:=[\gs]_0^{n+1}$. The embedding simply multiplies the first $n$ coordinates by $29$ and sets the last one to $0$%
\footnote{Formally., let $L=(l(0),l(1),...,l(T))$ be the line we obtained in the previous sub-step. 
In particular, note that $l(t)\in \{0,1\}^n$ for all $t$, $l(0)=\0_n$ is our starting point, and $d_1(l(t),l(t-1))=1$. 
A vertex $l(t)$ is embedded into $l'(t)=(\gs l(t),0)\in G_{n+1}$. An edge $(l(t),l(t+1))$ is embedded into the shortest path of length $\gs$ from $l'(t)$ to $l'(t+1)$, where all coordinates except one remain fixed.}. 
We denote by $E(t)\subset G_{n+1}$ the embedded edge from $l'(t)$ to $l'(t+1)$, which is simply the set of $\gsp$ grid points along the shortest path from $l'(t)$ to $l'(t+1)$. The \emph{direction} of $E(t)$ is $l(t)-l(t-1)=\frac{1}{\gs} (l'(t+1)-l'(t))\in \{\pm e_1,...,\pm e_{n+1}\}$, where $e_i$ are the unit vectors. We add one additional initial edge to the path $E(0):=\{(\0_n,\gs),(\0_n,\gsm),...,(\0_n,0)\}$ with direction $-e_{n+1}$. We denote by $L'=\cup_{t=0}^m E(t)$ the embedded line. We denote $l'(-1):=(\0_n,\gs)$ to be the \emph{origin} of the line $L'$.

\subsection{The potential function}

Now we define a potential function $\phi:G_{n+1} \ra \N$ along with some terminology that will be useful in the proofs. Figure~\ref{fig:2d} visualizes some of the important properties of the potential (albeit with the usual caveats that come with drawing an $n$-dimensional potential on a two-dimensional paper).

\begin{figure}[t]
\caption{Defining the potential function near a line/corner}\label{fig:2d}
\vspace{0.6cm}
\begin{tikzpicture}[xscale=0.5,yscale=0.5]
\draw (0,-3) -- (0,14) -- (29,14) -- (29,-3);
\draw[dashed] (0,-3)--(0,-6);
\draw[dashed] (29,-3)--(29,-6);

\filldraw (0,14) circle(0.15); 
\filldraw (29,14) circle(0.15); 
\node[left] at (0,14) {$l'(t-1)$};
\node[above] at (29,14) {$l'(t)$};

\node[above] at (14.5,14) {$E(t)$};
\node[right] at (29,7) {$E(t+1)$};

\draw[->] (0,14.5) -- (1,14.5);
\draw[->] (2,14.5) -- (3,14.5);
\node[right] at (3,14.5) {$- e_i$};
\draw[->] (6,14.5) -- (7,14.5);
\draw[->] (10,14.5) -- (11,14.5);
\draw[->] (19,14.5) -- (20,14.5);
\draw[->] (23,14.5) -- (24,14.5);
\draw[->] (27,14.5) -- (28,14.5);

\draw[->] (29.5,14) -- (29.5,13);
\draw[->] (29.5,11.5) -- (29.5,10.5);
\node[below] at (29.5,13) {$e_j$};
\draw[->] (29.5,3.5) -- (29.5,2.5);
\draw[->] (29.5,0) -- (29.5,-1);

\draw (0,10)--(25,10)--(25,14);
\node at (14,12) {(close; $t$)};

\draw[->] (2,11.5)--(2,12.5);
\draw[->] (2,11.5)--(1,11.5);

\draw[->] (6,11.5)--(6,12.5);
\draw[->] (6,11.5)--(5,11.5);

\draw[->] (10,11.5)--(10,12.5);
\draw[->] (10,11.5)--(9,11.5);

\draw[->] (23,11.5)--(23,12.5);
\draw[->] (23,11.5)--(22,11.5);

\draw[->] (19,11.5)--(19,12.5);
\draw[->] (19,11.5)--(18,11.5);

\draw (25,10)--(25,-3);
\draw[dashed] (25,-3)--(25,-6);

\node[rotate=270] at (27,4) {(close; $t+1$)};
\draw[->] (27,11.5)--(27,12.5);
\draw[->] (27,11.5)--(28,11.5);

\draw[->] (27,7.5)--(27,8.5);
\draw[->] (27,7.5)--(28,7.5);

\draw[->] (27,0)--(27,1);
\draw[->] (27,0)--(28,0);

\draw (4,10)--(4,-3);
\draw[dashed] (4,-3)--(4,-6);
\node[rotate=90] at (2,3) {(close; $t-1$)};

\draw[->] (2,7.5)--(2,6.5);
\draw[->] (2,7.5)--(1,7.5);

\draw[->] (2,0)--(2,-1);
\draw[->] (2,0)--(1,0);

\draw (12,10)--(12,6)--(25,6);
\node at (14.5,8) {(semi-close; $t$)};

\draw[->] (23,7.5)--(23,8.5);
\draw[->] (23,7.5)--(22,7.5);

\draw[->] (19,7.5)--(19,8.5);
\draw[->] (19,7.5)--(18,7.5);

\draw (21,2)--(21,-3);
\draw[dashed] (21,-3)--(21,-6);
\node[rotate=270] at (23,-3) {(semi-close; $t+1$)};

\draw[->] (23,0)--(23,1);
\draw[->] (23,0)--(24,0);

\draw (8,10)--(8,-3);
\draw[dashed] (8,-3)--(8,-6);
\node[rotate=90] at (6,3) {(semi-close; $t-1$)};

\draw[->] (6,7.5)--(6,6.5);
\draw[->] (6,7.5)--(5,7.5);

\draw[->] (6,0)--(6,-1);
\draw[->] (6,0)--(5,0);

\draw (12,6)--(12,2)--(25,2);
\node[text centered, text width=3.5cm] at (14,4) {(semi-far; $t$)};

\draw[->] (19,3.5)--(19,4.5);
\draw[->] (19,3.5)--(18,3.5);

\draw[->] (23,3.5)--(23,4.5);
\draw[->] (23,3.5)--(22,3.5);

\draw (12,2)--(12,-3);
\draw[dashed] (12,-3)--(12,-6);
\node[rotate=90] at (10,3) {(semi-far; $t-1$)};

\draw[->] (10,0)--(10,-1);
\draw[->] (10,0)--(9,0);

\draw (17,2)--(17,-3);
\draw[dashed] (17,-3)--(17,-6);

\draw[->] (10,7.5)--(10,6.5);
\draw[->] (10,7.5)--(9,7.5);

\node[rotate=270] at (19,-3) {(semi-far; $t+1$)};

\draw[->] (19,0)--(19,1);
\draw[->] (19,0)--(20,0);

\node at (14.5,0) {(far)};

\end{tikzpicture}
{\small The pair in the parenthesis in each region specifies the distance of points from the line and the edge index of the points, according to the definition of $\phi$. Arrows in each region indicate the directions (among $\pm e_i, \pm e_j$) in which the potential decreases. Arrows along the edges indicate the direction in which potential on the line decreases.}
\end{figure}
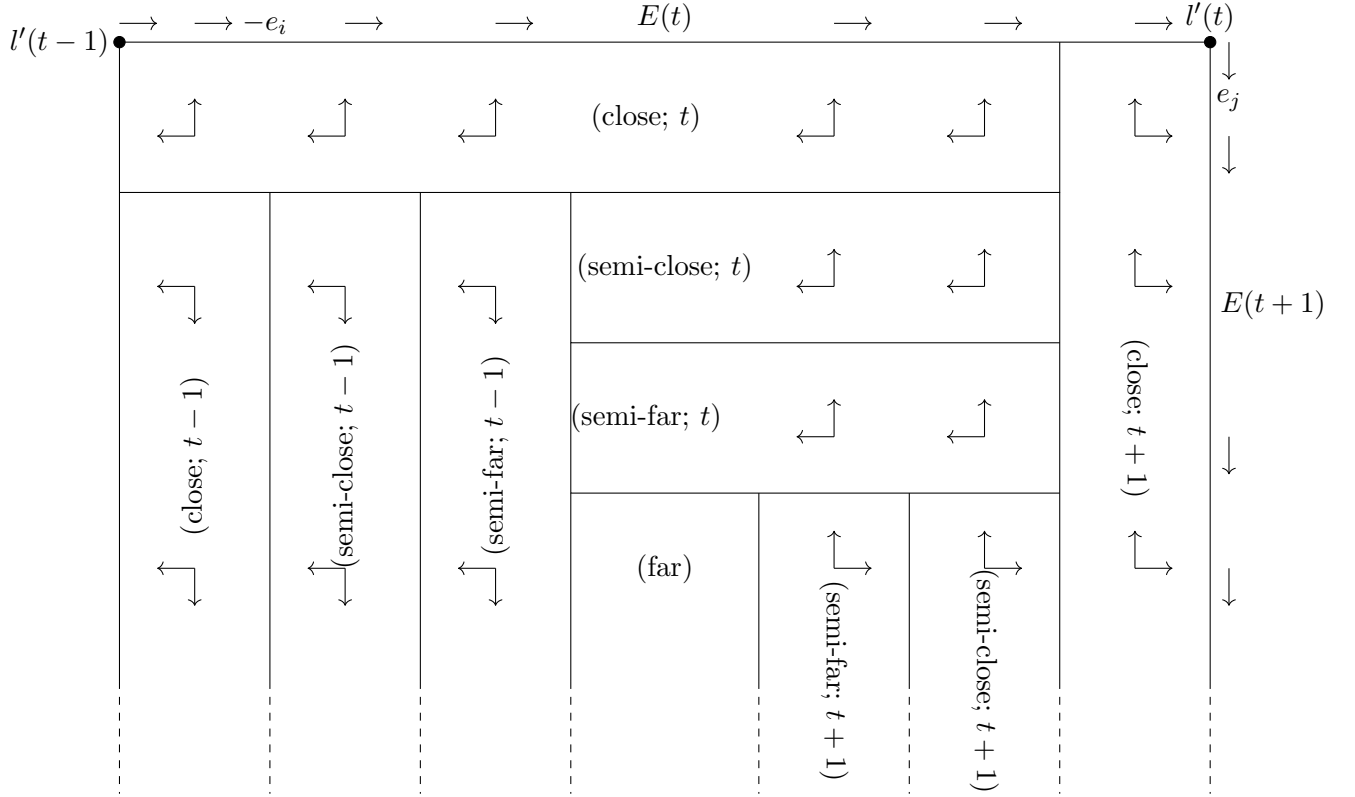

\paragraph{Potential for points on the line} For $x \in L'$, we denote by $t$ the maximal index such that $x\in E(t)$, and we set 
\begin{align}\label{eq:on}
\phi(x)=\thgs(T-t)+d_1(x,l'(t)).
\end{align}
Note that the term $88(T-t)$ captures the length of path that follows the line from $l'(t)$ till its end, multiplied by some constant.

\paragraph{Potential for close points}
A point $x\in B(L',4)\setminus L'$ is called \emph{close} to the line point. We set the \emph{edge-index} of $x$ to be the maximal index $t$ such that $x\in B(E(t),4)$, and we set 
\begin{align}\label{eq:close}
\phi(x)=\thgs(t+T+1)+d_1(x,l'(t-1)).
\end{align}
Note that the term $88(t+T+1)$ captures the length of path that starts at $l'(t-1)$ follows the line till its origin and then follows the line till its end, multiplied by some constant.

\paragraph{Potential for semi-close and semi-far points}
The \emph{edge-index} of $x\in B(L',12)\setminus B(L',4)$ is the minimal index $t$ such that $x\in B(E(t),12)$. 

If $5\leq d_\infty (x,E(t))\leq 8$ we call $x$ a \emph{semi-close point} and the potential is set similarly to the close points case 
\begin{align}\label{eq:semi-close}
\phi(x)=\thgs(t+T+2)+d_1(x,l'(t-1)).
\end{align}
The distinctions from the close point case are that index $t$ is chosen according to the minimal close index rather than maximal, and the additional constant of $\thgs$.

If $9\leq d_\infty (x,E(t))\leq 12$ we call $x$ a \emph{semi-far point} and we set 
\begin{align}\label{eq:semi-far}
\phi(x)=\thgs(t+T+2)+d_1(x,l'(t-1)) +(\tgs - 2x_{n+1}).
\end{align} 
The distinction from semi-close points is that the potential decreases with the special dimension $n+1$. Indeed when $x_{n+1}$ increases by 1 the distance $d_1(x,l'(t-1))$ increases by (at most) 1, but $\tgs - 2x_{n+1}$ decreases by 2.

\paragraph{Potential for points far from the line}
Finally, for $x\notin B(L',12)$, we set 
\begin{align}\label{eq:far}
\phi(x)=\thgs(T+1)+d_1(x,l'(-1)).
\end{align}
This definition is inspired by a path that goes from $x$ \emph{directly} to the origin and then follows the path till its end. The only difference is that the length of the path that goes from the origin till the end is multiplied by some constant.

This potential has two key properties: \emph{locality} and \emph{existence of a dominating direction for small cubes}. In the following subsections, we define these properties formally and show that indeed $\phi$ satisfies them.

\subsection{Locality} For every $x\in G_{n+1}$ we denote by $r_i(x_i):=\gs\cdot \1_{x_i\geq \hgs}$ and $r(x)=(r_i(x_i))_{i\in [n+1]}$ \emph{the rounding} of $x$ to the corners of $G_{n+1}$. For a corner point $y\in \{0,\gs\}^{n+1}$ the \emph{local information about the line} consists of whether the line $L'$ goes through $y$, and if so also in which stage it passes $y$ (i.e., the value of $t$ for $l'(t)=y$), and also the previous and the next vertices $l'(t-1)$ and $l'(t+1)$.

A potential satisfies the locality property if $\phi(x)$ can be calculated from the local information of $r(x)$. Indeed, by accessing the information $(l'(t'-1),l'(t'),l'(t'+1))$ at the point $r(x)$ one can determine whether $x$ is on the line, close, semi-close, semi-far, or far from the line. Moreover, the index $t'$ 
can be deduced from the point $l(t')$ because the line is metered. The maximal and minimal indices $t$ such that $x$ is on / close / semi-close / semi-far to $E(t)$ can also be deduced from $(t',l'(t'-1),l'(t'),l'(t'+1))$. Finally, in all the equations defining $\phi(x)$ (see Equations \eqref{eq:on},\eqref{eq:close},\eqref{eq:semi-close},\eqref{eq:semi-far} and \eqref{eq:far}), the potential depends only on $t$ and $l'(t-1),l'(t) \in \{l'(t'-1),l'(t'),l'(t'+1)\}$.

\subsection{Existence of a dominating direction for small cubes}\label{sec:domdir} 
A \emph{small cube} 
is defined by $C(\uv,\ov):=\{x\in G_{n-1}: \uv_i\leq x_i \leq \ov_i \ \forall i\in [n+1]\}$ where $\uv,\ov\in G_{n+1}$, $\uv_i\leq \ov_i$ and $d_\infty (\uv,\ov)\leq 4$. The corners $\uv$ and $\ov$ are the lower and upper corners of the cube. The cube is called small because $d_\infty (\uv,\ov)\leq 4$. 
We have $2(n+1)$ possible directions $\{\pm e_1,\pm e_2,...,\pm e_{n+1} \}$.

We begin with the weaker notion of feasible direction.
\begin{definition}[Feasible direction]\label{def:feasible} \hfill

Given a small cube $C(\uv,\ov)$ we say that $e_i$ is a \emph{feasible direction} if there exists some $x\in C(\uv,\ov)$ such that $x+e_i \in G_{n+1}$.
The case of $-e_i$ as a feasible direction is defined analogously.
\end{definition}

We now define the notion of dominating direction, which formalizes are desiderata for points that are far from the embedded end of the line.
\begin{definition}[Dominating direction]\label{def:dd} \hfill

Given a potential $\phi:G_{n+1}\ra \Real$ and a small cube $C(\uv,\ov)$ we say that $e_i$ is a \emph{dominating direction} if it is a feasible direction, and for every $x\in C(\uv,\ov)$ that satisfies $x+e_i \in G_{n+1}$ we have $\phi(x)> \phi(x+e_i)$.
The case of $-e_i$ as a dominating direction is defined analogously.

\end{definition}
Simply speaking, the notion of dominating direction captures the idea that ``moving from any point in the small cube in the direction $e_i$ decreases the potential". Note that moving in a given direction might be impossible (infeasible) if we are located on the boundaries of the big cube $G_{n+1}$. However, we want to avoid an undesirable situation where a direction is defined to be dominating simply because for all points in the small cube we cannot move in this direction, hence we require the direction to be feasible. 

The second key property of the potential $\phi$ defined above is the following.
\begin{proposition}\label{pro:dom-dir}
Every small cube $C(\uv,\ov)$ has a dominating direction except for the zero-dimensional cube that is the end-of-line $C(l'(T),l'(T))$.  
\end{proposition}

\begin{proof}
First we observe that all points in a small cube (with edge length 4) can be $12$-close (in the $d_\infty$ distance) to at most two adjacent edges $E(t)$ and $E(t+1)$. This simply follows from the fact that $12+12+4<\gs$. 

The cases that we consider are as follows:

\begin{itemize}
\item \textbf{Case A} All points in $C(\uv,\ov)$ are $12$-close to two different edges.
\item \textbf{Case B} All points in $C(\uv,\ov)$ are $12$-close to the line and $C(\uv,\ov)$ admits a point that is close to a single edge.
\item \textbf{Case C} The cube $C(\uv,\ov)$ admits a point that is $13$-far from the line. 

\end{itemize}

We first prove the lemma for (the most intricate) Case A.

\paragraph{Case A} We recall that $l'(t)=E(t) \cap E(t+1)$. We denote by $\pm e_i$ the direction of $E(t)$. We denote by $\pm e_j$ the direction of $E(t+1)$.
We assume without loss of generality that $l'_i(t)=l'_j(t)=0$. Equivalently we assume (w.l.o.g) that the direction of $E(t)$ is $-e_i$ and the direction of $E(t+1)$ is $e_j$. The three remaining cases ($(l'_i(t),l'_j(t))\in \{0,\gs\}^2$) follow by analogous arguments.

We denote by $c=\arg \max_{x\in C(\uv,\ov)} d_1(x,l'(t))$ the corner of $C(\uv,\ov)$ that is most far from $l'(t)$. Note that since we use $d_1$ distance this corner is unique. We denote by $I:=\{k\in [n+1] :c_k \neq l'_k(t) \}$ the subset of indexes for which the direction of the cube $C(\uv,\ov)$ toward the point $l'(t)$ is feasible according to Definition \ref{def:feasible}. Note also that \emph{all} directions from $l'(t)$ toward the interior of $G_{n+1}$ are feasible (since $d_\infty(c,l'(t))\leq 12<\gs$). 

Now we split to cases according to the set $I$ and according to the values of $d_\infty (c,E(t))$ and $d_\infty (c,E(t+1))$; these values indicate the location of $c$ with respect to the edges.

Figure \ref{fig:2d} is two-dimensional, whereas the actual path is of dimension $(n+1)$. However, the first three cases below (Cases A.1-A.3) ``reduce the proof to this two-dimensional picture" by showing that if the cube has additional dimensions that are out of this planar picture then a dominating direction can be easily found.

\paragraph{Case A.1: $I\setminus \{i,j,n+1\}\neq \emptyset$.} In such a case there is a feasible direction whose index is $k\notin \{i,j,n+1\}$ that gets closer to $l'(t)$. Without loss of generality we assume that this direction is $-e_k$ (or equivalently we assume that $l'_k(t)=0$). We will argue that $-e_k$ is a dominating direction. 
Let $x\in C(\uv,\ov)$ be such that $x_k=\ov_k>0$.
When moving from $x$ to $x-e_k$, the $d_1$ distance to both $l'(t)$ and $l'(t-1)$ decreases by 1; namely $d_1(x-e_k,l'(t))=d_1(x,l'(t))-1$ and $d_1(x-e_k,l'(t-1))=d_1(x,l'(t-1))-1$. Also, it is easy to verify that the edge-index of $x$ and of $x-e_k$ is identical. Therefore, if $x$ and $x-e_k$ are both close (respectively both semi-close or both semi-far), the the potential decreases. Altevnatively, if by moving from $x$ to $x-e_k$ we jump from semi-far to semi-close (respectively from semi-close to close, or from close point to a point on the line), while maintaining the edge index, the potential again decreases. 
Therefore $-e_k$ is indeed a dominating direction.\\

In the following two cases we deal with the case $n+1\in I$. We recall that by the embedding of the line $l'_{n+1}(t)=0$ (the case $t=-1$ is impossible because we have assumed existence of the edge $E(t)$).

\paragraph{Case A.2: $0<c_{n+1}\leq 8$.} Exactly the same arguments as in Case 1 can be applied with the direction $-e_{n+1}$.

\paragraph{Case A.3: $9\leq c_{n+1} \leq 12$.} In such a case all points in $C(\uv,\ov)$ are either semi-close or semi-far points. Thus, the edge-index of all points in $C(\uv,\ov)$ is\footnote{Here it is crucial that for both semi-close or semi-far points the edge index is set to be the minimal one.} $t$. Therefore, for the direction $e_i$ the distance from $l'(t-1)$ decreases and hence the potential decreases.\\

Now we remain with the case where $I\subset \{i,j\}$ which means that the cube contains only points of the form $\{x:\forall k\notin \{i,j\}, \;\; x_k=l'_k(t)\}$. Namely, we remain with the much easier two-dimensional problem presented in Figure \ref{fig:2d}.

\paragraph{Case A.4: $C(\uv,\ov)\cap E(t+1)\neq \emptyset$ and $i\in I$.} In such a case all points in $C(\uv,\ov)$ are either point on $E(t)$ or close points with edge-index\footnote{Note that here it is crucial that for close points we have set the index to be the \emph{maximal} one.} $t+1$. Therefore, we can choose the direction to be $-e_i$, which ensures that $d_1(x-e_i,l'(t))=d_1(x,l'(t))-1$, and it is easy to see (Equations \eqref{eq:on} and \eqref{eq:close}) that the potential decreases.\\

All other cases are quite trivial and the arguments follow from Figure \ref{fig:2d}. If $C(\uv,\ov)$ is a rectangle that is disjoint to $E(t+1)$, a dominating direction is $-e_j$. If it is contained in the line $E(t+1)$, a dominating direction is $e_j$. Otherwise, if it is contained in the line $E(t)$, a dominating direction is $e_i$. This completes the proof for Case A.

\paragraph{Case B.} We recall that in this case there exists a point $x\in C(\uv,\ov)$ that is $12$-close to an edge $E(t)$ with direction $\pm e_i$, and $x$ is $13$-far from all other edges. Without loss of generality we assume that the direction of $E(t)$ is $-e_i$ (this is also the case for the edge $E(0)$).

If $C(\uv,\ov)\cap E(t)=\emptyset$ we argue that $-e_i$ is a dominating direction. For every $x\in C(\uv,\ov)$ we have $d_\infty(x,E(t))=d_\infty(x+e_i,E(t))$ i.e., the distance from $E(t)$ is maintained. If the edge-index of $x+e_i$ remains $t$, the potential decreases because \footnote{Note that also for the case of $E(t)=E(0)$ whose direction is $-e_{n+1}$ the potential decreases in the direction $e_{n+1}$ even for semi-far points, where this direction has special meaning.} $d_1(x+e_i,l'(t-1)) = d_1(x,l'(t-1))-1$. The other possibility is that $x$ is semi-close or semi-far point and $x+e_i$ becomes a semi-far point with edge index $t-1$ (see Figure \ref{fig:2d}). In such a case we gain at most $\gs$ in the $d_1$ term (because now we count the distance from $l'(t-2)$ rather than from $l'(t-1)$) and we gain at most $\tgs$ in the term $\tgs - 2 x_{n+1}$, but we lose $\thgs>\tgs+\gs$ because the index has been decreased by 1.

If $C(\uv,\ov) \neq C(\uv,\ov)\setminus E(t)\neq \emptyset$, we pick a feasible direction $\pm e_k$ that gets closer to the edge $E(t)$. By similar arguments to the Case A.1, the potential decreases in this direction.

Finally if $C(\uv,\ov)\subset E(t)$ the dominating direction is $-e_i$.

\paragraph{Case C.} We recall that in this case there exists a point $x\in C(\uv,\ov)$ that is far from the line. Therefore, all $x\in C(\uv,\ov)$ are either far or semi-far from the line. If $e_{n+1}$ is a feasible direction then it is dominating, because it reduces the distance from the origin, see Equations \eqref{eq:semi-far}, \eqref{eq:far}. If $e_{n+1}$ is not feasible then any direction that moves toward the origin is dominating.  
\end{proof}

\section{Query Complexity bounds}\label{sec:qc}
The potential $\phi$ is a key ingredient in our lower bounds. However, without any additional ingredients, it is not sufficient to produce a reduction. Namely, if we consider the potential game with identical interest $\phi$ where each player is responsible for a single coordinate, there might be undesirable mixed Nash equilibria that are not ``located close to" the end-of-line. Below we present two different additional ingredients that ``concentrate" the support of the players in an equilibrium in a small cube, which allows us to use Proposition \ref{pro:dom-dir} to argue that all equilibria are supported near the end-of-line. One technique is 
replication, which will be utilized later to prove $2^{\poly(n)}$ lower bound on the communication complexity in $n$-player binary-action games.
The second technique is high degree imitation, which will be utilized later to prove $\poly(N)$ lower bound on the communication complexity in two-player $N$-action games.

\subsection{Replication}\label{sec:rep}

The first idea is to replicate the players that are responsible for a single coordinate and make the choice ``collective" by choosing the average of all numbers. The advantage of this approach is that it allows us to prove a lower bound for binary-action games. The disadvantage of this approach is that the lower bound is $2^{\Omega(\sqrt{n})}$ (rather than $2^{\Omega(n)}$).

Instead of having a single player $i$ that is responsible for the $i$-th coordinate we will have a team $m=\Theta(n)$ players that collectively choose the $i$-th coordinate.
Formally, the set of players is $\{(i,j)\}_{i\in [n],j\in [m]}$ the actions of all $nm$ players is binary $a_{i,j}\in\{0,\gs\}$. The actions of the $i$'th teams defines a number $\oa_i=\frac{1}{m}\sum_{j\in [m]} a_{i,j}\in [0,\gs]$. We denote by $\ophi:[0,\gs]^{n+1}\ra \Real$ the multilinear extension of $\ophi:[\gs]_0^{n+1}\ra \Real$ with respect to the closest integers. Namely, for $x=c+y=(c_1,...,c_{n+1})+(y_1,...,y_{n+1})\in [0,29]^{n+1}$ where $c_i\in \Z$ and $0\leq y_i <1$ we define 
\begin{align}\label{eq:multiLin}
\ophi (x)=\Ex_{s_i\sim Ber(y_i)} \phi(c_1+s_1,...,c_{n+1}+s_{n+1}),
\end{align}
where $Ber$ denotes the Bernoulli distribution.
The utility in the potential game is simply an identical interest utility that is given by $$u_{i,j}((a_{i,j})_{i,j})=u((a_{i,j})_{i,j})=-\ophi(\oa_1,...,\oa_{n+1}).$$

\begin{proposition}\label{theo:binary}
For sufficiently large $m=\Theta(n)$ the unique Nash equilibrium of the potential game $u(a)=-\ophi(\oa_1,...,\oa_{n+1})$ is a pure Nash equilibrium where $(\oa_1,...,\oa_{n+1})=l'(T)$ is the end-of-line of the line $L'$ (and in particular $a_{i,j}=l'_i(T)$). 
\end{proposition}

\begin{proof}
Let $\alpha=(\alpha_{i,j})_{i\in [n],j\in [m]}$ be a mixed Nash equilibrium. 
We choose $m=\Theta(n)$ that ensures $\Prob_{\alpha}(\oa_{i}- \frac{1}{m}\Ex[\sum_j \alpha_{i,j}]>0.1)\leq 3^{-n}$ for every $i,j$. Note that the choice of $m$ can be done independently of $\alpha$, by (the non i.i.d. variant of) Hoeffding inequality. Simply speaking, we set $m$ such that all averages will be $0.1$-concentrated around their expectation with very high probability of $1-3^{-n}$.

The point $\oa=(\oa_i)_{i\in [n+1]}\in [0,\gs]^{n+1}$ defines a small cube $C_{\oa}$ in $G_{n+1}$. If $\oa_i\in \{0,29\}$ we set $\uv_i=\ov_i=\oa_i$. Otherwise, if $\oa_i\in (0,29)$ we set $\uv_i= \max \{\lfloor \oa_i-0.2 \rfloor,0 \}$ and $ \ov_i = \min \{ \lceil \oa_i +0.2 \rceil,29 \}$. We set $C_{\oa}=C(\uv,\ov)$. Simply speaking, $C_{\oa}$ is the expansion of the (continuous) cube $\prod_i [\oa_i-0.2,\oa_i+0.2]$ to the integer grid, except for the special case where $a_i$ lands exactly on the boundary.
In such a case, in this dimension the cube has $0$ length (i.e., a segment consisting of a single point). 
We denote by $\oC_{\oa}:=\conv(C_{\oa})$ the continuous cube. 

By Proposition \ref{pro:dom-dir}, the cube $C_{\oa}$ admits a dominating direction $k\in [n+1]$. Without loss of generality we assume that this direction is $+e_k$. We argue that all players $(k,j)$ for $j\in [m]$ strictly prefer to play action $\gs$ rather than $0$. Since $\oa_k$ is not located on the boundary this implies that there is no equilibrium whose average is $\oa$.

For every point $x=c+y\in \oC_{\oa}$, with the notations of Equation \eqref{eq:multiLin}, such that $y_k\neq 0$ (i.e., $x_k\notin \Z$) we have 
\begin{align}\label{eq:d}
\begin{split}
\frac{\partial (-\ophi)}{\partial x_k}(x)&=\frac{\ophi(c_k,x_{-k}) - \ophi (c_k+1,x_{-k})}{1} \\
&=\Ex_{s_i \sim Ber(y_i)}\left[\phi(c_k,(c_i + s_i)_{i\neq k})-\phi(c_k+1,(c_i + s_i)_{i\neq k})\right]\geq 1,
\end{split}
\end{align}

where the first equation follows from the fact that $\ophi$ is linear over the segment $\conv ((c_k,x_{-k}),(c_k+1,x_{-k}))$. The second equation is by the definition of the multilinear extension, and the last inequality holds because this difference is at least 1 \emph{for every realization} of $s_i$ since $e_k$ is a dominating direction in the cube.

We denote $z^0\in [0,29]^{n+1}$ to be $z_i=\oa_i$ for $i\neq k$ and 
$z^0_k=\frac{1}{m}\sum_{l\neq j} a_{k,l}$. Namely $z^0$ is the random variable of $\oa$ in case player $(k,j)$ plays $0$. We denote $z^{\gs}=z^0+\frac{29}{m}e_k$ the random variable of $\oa$ in case player $(k,j)$ plays $\gs$. The bound on the derivative (Equation \eqref{eq:d}) implies that for $z^0,z^{\gs}\in \oC_{\oa}$ we have
\begin{align}\label{eq:dif}
\ophi(z^0)-\ophi(z^{\gs})\geq \frac{29}{m} \min_{x\in \conv(\{z^0,z^{\gs} \})} \frac{\partial (-\ophi)}{\partial x_k}(x) \geq \frac{29}{m}
\end{align}
By deviation from $0$ to $\gs$ of player $(k,j)$ the (identical) expected utility increases by at least
\begin{align*}
\Ex_{\alpha} [\ophi(z^0)-\ophi(z^{\gs})] & = \Ex_\alpha [\ophi(z^0)-\ophi(z^{\gs})| \{z^0,z^{\gs}\}\subset \oC_{\oa}] \cdot \Prob_\alpha [\{z^0,z^{\gs}\}\subset \oC_{\oa}]\\
& \hspace*{4mm} + \Ex_\alpha [\ophi(z^0)-\ophi(z^{\gs})| \{z^0,z^{\gs}\}\not\subset \oC_{\oa}] \cdot \Prob_\alpha [\{z^0,z^{\gs}\}\not\subset \oC_{\oa}] \\
& \geq  \frac{\gs}{m}\Prob_\alpha [\{z^0,z^{\gs}\}\subset \oC_{\oa}] -200\cdot 2^n \Prob_\alpha [\{z^0,z^{\gs}\}\not\subset \oC_{\oa}] \\
& \geq  \frac{\gs}{m} (1-(n+1)3^{-n})-200\cdot 2^n (n+1)3^{-n} \geq \frac{\gsm}{m}>0,
\end{align*}
where the second inequality follows from Equation \eqref{eq:dif} and the fact that the potential is bounded by $200\cdot 2^n$. The third inequality follows from the choice of $m$ the union bound and the fact that if $\oa_k$ is $0.1$-close to its expectation, then necessarily both $z^0_k$ and $z^{\gs}_k$ are $0.2$-close to this expectation.  
\end{proof}

\begin{corollary}\label{cor:bin}
The query complexity of finding a Nash equilibrium (possibly mixed) in $n$-player binary-action potential games is at least $2^{\Omega(\sqrt{n})}$.
\end{corollary}

\begin{proof}
By the locality of $\phi$, we can reduce the End-of-Line problem over the $n$-dimensional hypercube to the mixed Nash equilibrium problem with $nm=\Theta(n^2)$ players using the game of Proposition \ref{theo:binary}.
\end{proof}

\subsection{High Degree Imitation}\label{sec:hd}
The second idea is to define a game where players collectively choose \emph{two} points $a$ and $b$ (rather than one), and are incentivized to choose close-by points by
a \emph{high degree imitation} utility. The advantage of this approach is that it will yield a tight $2^{\Omega(n)}$ bound. The disadvantage of this approach is that it proves a lower bound for games with constant but large ($\gsp$) number of actions.

Given a line $L$ over the hypercube we define a $2(n+1)$-player $\gsp$-action potential game.

Every player $i\in [n+1]$ in Alice's team chooses a coordinate $a_i \in [\gs]_0$. The profile of Alice's team is denoted by $a\in G_{n+1}$. 
Every player $i\in [n+1]$ in Bob's team chooses a coordinate $b_i \in [\gs]_0$. The profile of Bob's team is denoted by $b\in G_{n+1}$.

All players in Alice's team have identical utility that is given by 
\begin{align*}
u_A(a,b)=-\sum_{i=1}^{n+1} (a_i-b_i)^{2n}-2\phi(a)
\end{align*}
All players in Bob's team have identical utility that is given by 
\begin{align*}
u_B(a,b)=-\sum_{i=1}^{n+1} (a_i-b_i)^{2n}-2\phi(b)
\end{align*}
We call the first term in the utilities of the players \emph{the imitation loss}, and the second term \emph{the potential loss}.
It is easy to verify that the game $(u_A,u_B)$ is a potential game whose potential is given by $\varphi(a,b)=-\sum_{i=1}^{n+1} (a_i-b_i)^{2n} -2\phi(a)-2\phi(b)$.

The following lemma states that all points that are played in a mixed equilibrium with positive probability are located close to each other (within a distance $4$ from each other in the $d_\infty$ distance).

\begin{lemma}\label{lem:cube-bound}
Let $(\alpha,\beta)$ be a mixed Nash equilibrium of the game potential game $(u_A,u_B)$, then there exists a small cube $C(\uv,\ov)$ (with edge-length 4) such that $\supp(\alpha)\cup \supp(\beta)\subset C(\uv,\ov)$.
\end{lemma}

The key ingredient for proving this lemma is the following property of high-degree loss functions.

\begin{lemma}\label{lem:hd}
Let $\beta\in \Delta([29]_0)$ be a distribution over the integers in $[29]_0$. Let $v^*=\min_{a\in [29]_0} \Ex_{b\sim \beta} (a-b)^{2n}$ be the minimal expected loss that can obtained be choosing an integer in $[29]_0$ against the distribution $\beta$. There exist an integer segment $[c,c+2]$ for $c\in \Z$ such that for every $a\notin [c,c+2]$, $a\in \Z$ we have $ \Ex_{b\sim \beta} (a-b)^{2n}\geq v^* + 3^n$.
\end{lemma}

Lemma~\ref{lem:hd} states that by choosing an integer out of the range of some ``best three consecutive integers" the high degree imitation loss increases by a large exponential term of at least $3^n$. The proof of this lemma is relegated to Appendix \ref{ap:hdlemma}. The challenging part in proving this lemma is for the case where the distribution $\beta$ admits exponentially small weights on some integers. Below we show how to utilize this lemma to prove Lemma \ref{lem:cube-bound}.

\begin{proof}[Proof of Lemma \ref{lem:cube-bound}]
We first show that $\supp(\alpha)$ is concentrated in a cube of edge length 2.

Let $\beta_i\in \Delta([29]_0)$ be the mixed action of player $i$ in Bob's team. We argue that the support of player $i$ in Alice's team is contained in $[c_i,c_i+2]$, when $[c_i,c_i+2]$ is the segment from Lemma \ref{lem:hd}.

Denote by $M_{-i}:=\Ex_{a_{-i}\sim \alpha_{-i}, b_{-i}\sim \beta_{-i}} \sum_{j\neq i} (a_j-b_j)^{2n}$ the imitation loss that is caused by players that are responsible for coordinates other than $i$. This term is not affected by the choice of $a_i$. 
Let $a_i^*\in [\gs]_0$ be the integer action that minimizes expected imitation loss with respect to $\beta_i$ (and ignores the potential loss) in the $i$th coordinate. The action $a_i^*$ yields a utility of $-M_{-i}-v^*-O(2^n)$. By Lemma \ref{lem:hd} any action $a_i \notin [c_i,c_i+2]$ yields a utility of at most $-M_{-i}-v^*-3^n$. Hence, $a^*_i$ is strictly better than $a_i$, and in particular $a_i$ does not belong to the support of an equilibrium.

Now we consider player $i$ in Bob's team. We recall that $M_{-i}$ denotes the imitation loss of players other than $i$ (the imitation loss is identical for Alice and for Bob). 
By playing $b_i=c_i+1$ Bob guarantees a loss of at most $-M_{-i}-1-\Theta(2^n)$ because the imitation loss is at most 1. At least one of the actions $a_i=c_i,c_i+1,c_i+2$ has a weight of at least $\frac{1}{3}$. Without loss of generality this action is $a_i=c_i$. By playing an action $b_i\notin [c_i-2,c_i+2]$ Bob's loss will be at least $-M_{-i}-\frac{1}{3}2^{2n} $ which is strictly worse than $-M_{-i}-1-\Theta(2^n)$. Therefore, $\supp(\alpha_i)\cup \supp(\beta_i)\subset [c_i-2,c_i+2]$.
\end{proof}

\begin{proposition}\label{theo:const-actions}
The potential game $(u_A,u_B)$ has a unique Nash equilibrium (including mixed ones). This is the pure Nash equilibrium where both players choose the end-of-line vertex $a=b=l'(T)$.   
\end{proposition}

\begin{proof}
By Lemma \ref{lem:cube-bound}, let $C(\uv,\ov)$ be a small cube that contains $\supp(\alpha)\cup \supp(\beta)$. Without loss of generality we assume that every facet of $C(\uv,\ov)$ contains at least one point of $\supp(\alpha)\cup \supp(\beta)$, because otherwise we can reduce the cube size in this dimension.

By Proposition \ref{pro:dom-dir} let $\pm e_i$ be a dominating direction of $C(\uv,\ov)$. W.l.o.g., we assume this direction is $e_i$. One of the players has actions in his support that appear in the $x_i=\uv_i)$ facet of the cube. W.l.o.g., we assume that this player is Alice's team. Namely, $\supp(\alpha)\cap \{x\in C(\uv,\ov):x_i=\uv_i \}\neq \emptyset$.

We argue that player $i$ in Alice's team gets (strictly) better payoff by playing $a_i=\uv_i+1$ rather than $a_i=\uv_i$ (which will contradict  the fact that $a$ is in the support of an equilibrium).

Since the direction $e_i$ is dominating, and potential values are integers, the term $-2\phi(a)$ increases by at least 2.

Since Bob's actions are contained in $C(\uv,\ov)$, we know that \emph{all} actions in the support of Bob's $i$-player satisfy $b_i\geq \uv_i$. 
The worst case with respect to the term $(a_i-b_i)^{4n}$ is the case where $b_i=\uv_i$ with probability 1 (otherwise, Alice's action $\uv_i+1$ "gets closer" to Bob's actions). In such a case the term $(a_i-b_i)^{4n}$ increases by exactly 1. Namely, Alice gains at least 2 in the potential term, and loses at most 1 in the imitation term.
\end{proof}

\begin{corollary}\label{cor:const}
The query complexity of finding a Nash equilibrium (possibly mixed) in $n$-player $30$-action potential games is at least $2^{\Omega(n)}$.
\end{corollary}

\begin{proof}
By the locality of $\phi$, we can reduce the End-of-Line problem over the hypercube to the mixed Nash equilibrium problem using the game of Proposition \ref{theo:const-actions}.
\end{proof}

\section{Lifting to Communication}\label{sec:cc}

We first define formally the communication complexity problems, and state formally our main results.

\paragraph{Two-player communication} The communication problem of (possibly mixed) Nash equilibrium $\textsc{2-Nash-Potential}(N)$ in two-player potential games is a promise problem with private inputs $u_A$ for Alice and $u_B$ for Bob, when $|A|=|B|=N$ and it is promised that the pair $(u_A,u_B)$ is a potential game. The output is a mixed Nash equilibrium $(\alpha,\beta)\in \Delta(A)\times \Delta(B)$.

\paragraph{$n$-player communication} The communication problem of (possibly mixed) Nash equilibrium $\textsc{Multiplayer-Nash-Potential}(n)$ in $n$-player binary action potential games is a promise problem with private inputs $(u_1,...,u_{n/2})$ for Alice and $(u_{n/2+1},...,u_n)$ for Bob, where $u_i:\{0,1\}^n \ra \Real$ is the utility of player $i$. It is promised that the tuple $(u_1,u_2,...,u_n)$ is a potential game. The output is a mixed Nash equilibrium $(x_1,...,x_n)$ where $x_i \in \Delta(\{0,1\})$.

\begin{remark}[$n$-party vs two-party communication]
In $n$-player games it is also natural to study the $n$-party communication problem where each player holds as a private input his own utility function. This problem is clearly harder than the described above $\textsc{Multiplayer-Nash-Potential}(n)$ where the input of the utility functions is distributed only between two parties. Since our result is negative it obviously applies also to the harder $n$-party communication problem.
\end{remark}

Our main results show hardness of $\textsc{2-Nash-Potential}(N)$ and of $\textsc{Multiplayer-Nash-Potential}(n)$.

\begin{theorem}\label{theo:2p}
There exists a constant $c>0$ such that $CC(\textsc{2-Nash-Potential}(N))\geq N^c$.
\end{theorem}

\begin{theorem}\label{theo:np}
$CC(\textsc{Multiplayer-Nash-Potential}(n))\geq 2^{\Omega(\sqrt{n})}$.
\end{theorem}

These theorems are proved in the following two subsections.

\subsection{$n$-player hardness}\label{sec:n-pl}

We lift the query results of Section \ref{sec:qc} to a communication model using the result of \cite{HN12,GP14}, Theorem \ref{thoe:ccEOL}. 
The high-level idea is to use the replication technique (Section \ref{sec:rep}). However, unlike the replication game in the query model, here we want \emph{both} Alice's and Bob's teams to choose an action. To ensure that both teams will choose close-by points we use a simple imitation gadget (simpler than the high-degree imitation). In the constructed game the two  teams of agents try to imitate each other and in parallel try to maximize their own potential. Note that Theorem \ref{thoe:ccEOL} refers to a case where the input about the line is distributed between Alice and Bob, and therefore neither of the teams know what the potential is. We add an additional ingredient to the game. Every team in addition to a chosen point, also reports her local information about the line at the chosen point. Here we utilize the locality property of the potential, which implies that every team can report only a small number of bits. We incentivize teams to report their local information truthfully, and we set the utility from the potential to be the potential \emph{with respect to the reported local information}\footnote{Similar technique has been applied in \cite{CCNash,GR18,CCPLS}.}. Now the utilities of the teams indeed depend only on their private information in the End-of-Line problem of Theorem \ref{thoe:ccEOL}.

\paragraph{An obstacle.} Even without defining formally the notion of \emph{local information about the line} (which is done formally below), it is clear that this local information varies from point to point. More concretely, the cube is partitioned into disjoint subsets $(B_i)$ where the local information differs. On the boundaries of such two sets $B_i$ and $B_j$ when agents report truthfully the local information of $B_i$ the players that choose the point \emph{do not} have an incentive to deviate from $B_i$ to $B_j$, even if it has higher potential, because then the report at the point $B_j$ will turn out to be false.\footnote{Similar obstacle has raised in \cite{CCPLS}. However, their technique to resolve this problematic issue cannot be applied here, because for mixed Nash equilibria analysis their technique is not valid.} This may create spurious equilibria located close to these boundaries (and hence far from the end-of-line).

\paragraph{A solution to the obstacle.} We resolve this problematic issue by introducing \emph{two teams} of agents that report local information of every team. We set \emph{different} notions of local information for the two teams. These notions of locality induce two different partitions into sets where local information changes $(B_i)$ and $(C_i)$. The notions of the locality are set in a way that the boundaries of the $(B_i)$s are disjoint to the boundaries of the $(C_i)$s. Which one of the reports do we use in order to evaluate the potential? It depends on the location of the point. For points that are close to the boundaries of the $(B_i)$s we use the $C_i$ valuation of potential, and conversely, for points that are close to the boundaries of the $(C_i)$s we use the $B_i$ valuation of potential.

In addition to the above modification for the players who report the local information, our construction in the communication model uses \emph{both} ingredients that were presented above in the query model: imitation (Section \ref{sec:hd}) and replication (Section \ref{sec:rep}). The roles of these ingredients are as follows. Imitation is used in order to guarantee that both teams will choose close-by points and hence will report the local information over the same region of the line. Replication is used to guarantee that unilateral deviation of a player will not change significantly the chosen point, and hence the local information will remain fixed. Saying it differently, without the replication ingredient unilateral deviation of a player might change the local information too drastically such that neither of the two teams will report the local information with respect to the correct point. Without the imitation ingredient, we just have a single team that chooses a point. However, in the communication model, we want the notion of locality for each player to depend on the point of \emph{its own} team (this is needed for the potential property of the game). This requires a game where both teams choose points, and we would like these points to be close.

\subsubsection{Local information, its reporting, and report-based potential}\label{sec:localinfo}

In this subsection, we describe formally the reporting of local information and the potential computation given the reports. 

We recall that an information of a \emph{vertex} $v$ in the line is a triple of $tsp_v:=(t_v,s_v,p_v)\in \{0,1\}^3$ that indicates whether the line goes through $v$, and if so also indicates the successor (among the two possible) and the predecessor (among the two possible).
In the communication problem of Theorem \ref{thoe:ccEOL} Alice holds an array of size 3 for each bit. Bob holds an index in $[3]$ for each bit. Here we slightly expend the notion of local information. For a vertex $v$ we include in the \emph{local information of} $v$ the $(n+1)$ triples $tsp_v$ of all the neighbours of the vertex $v$ in the cube.

Let $r_1,r_2$ be the two \emph{reporter teams} of a player (Alice or Bob). Each reporter team of Alice will essentially report $3(n+1)$ arrays of size 3, one for each bit. Similarly, each reporter team of Bob will essentially report $3(n+1)$ indexes of size 3.

\paragraph{Local information of the reporters, and reference vertex} For every point $x\in [0,\gs]^{n+1}$ that is $13$-close to a vertex i.e., there exists a $v\in \{0,\gs\}^{n+1}$ such that $d_\infty (v,x)\leq 13$, the local information at $x$ is define by the local information of $v$, for both teams $r_i$ $i=1,2$. Note that if a $13$-close vertex exists it is unique.

For a point $x$ that is $13$-far from all vertices but is $13$-close to an edge $(v,v+\gs e_i)$ where $v\in \{0,\gs\}^{n+1}$ and $v_i=0$, we note that this edge is unique. 
Here we associate the local information differently for the teams $r_1$ and $r_2$. 
For the team $r_1$, if $x_i \in (13,\gs-14]$, then the local information of $x$ is the local information of $v$. Otherwise (if $x_i \in (\gs-14,\gs-13)$) the local information of $x$ is the local information of $v+\gs e_i$. The \emph{reference boundary of the team $r_1$} is denoted by $V_1$ and is defined to be the set of all points $x$ that are $13$-close to an edge $(v,v+\gs e_i)$ such that $x_i=\gs-14$. Namely, $V_1$ is the set of points where the reference to a vertex changes for team $r_1$.

For the team $r_2$, if $x_i \in (13,14]$, then the local information of $x$ is the local information of $v$. Otherwise (if $x_i \in (14,\gs-13)$) the local information of $x$ is the local information of $v+\gs e_i$. The \emph{reference boundary of the team $r_2$} is denoted by $V_2$ and is defined to be the set of all points $x$ that are $12$-close to an edge $(v,v+\gs e_i)$ such that $x_i=14$. Namely, $V_2$ is the set of points where the reference to a vertex changes for team $r_2$.

Finally, for points that are $12$-far from all edges we do not need to define the notion of local information. The potential in all these points is defined independently of the line.

\paragraph{Relevance of reports} Both Alice's and Bob's team use the reported information. Given the fact that we have two reporters (in both Alice's and Bob's team), we shall define formally whose  information do we use, in parallel we define the notion of \emph{relevance of a report} that, simply speaking means that this report is used in the calculations of the potential. 

For $x$ that is $13$-close to a vertex $v$ we use the report $r_1$ and say that $r_1$ is \emph{relevant}. Since $r_2$ is not used we say that $r_2$ is \emph{irrelevant}.

For $x$ that is $13$-far from all vertices but $13$-close to an edge $(v,v+\gs e_i)$ for $i\neq n+1$ use the report $r_1$ and say that is \emph{relevant} iff $x_i\leq \gs/2$. Similarly we use the report $r_2$ and say that it is \emph{relevant} iff $x_i> \gs/2$.
For the special initial edge $((\0_n,\gs),(\0_n,0))$ (namely when $i=n+1$) we do not need any information about the line, hence we do not use any of the reports and we say that both reports are \emph{irrelevant}. 

For $x$ that is $13$-far from all edges, again we do not need any information about the line and we say that both reports are \emph{irrelevant}. 

We denote by $R_i\subset [0,\gs]^{n+1}$ $i=1,2$ the set of all points where report $r_i$ is relevant.

The \emph{relevance boundary} is defined to be the set of all points that contain arbitrary close pair of points $x^1,x^2$ such that $r_i$ is relevant at $x^i$. The relevance is the union of the set of points that are within a distance of exactly $13$ from a vertex $v$ and the set of points $x$ that are $13$-close to an edge $(v,v+\gs e_i)$ and such that $x_i=\frac{\gs}{2}$.  

\paragraph{Report-based potential} Given reports of all teams $r_1^A,r_2^A,r_1^B$, and $r_2^B$ and realized points by the teams $\oa$ and $\ob$ we use the following procedure to deduce the potential. First, we identify the relevant report of each team. This report is about a vertex $v(\oa)$, $v(\ob)$ and their neighbours in the cube. Second, we combine Alice's information at her $(n+1)$ vertices with Bob's information at his $(n+1)$ vertices. In principle, these sets of vertices may not overlap; However, we will see that in an equilibrium it does not happen. The intersection of these sets of vertices is called \emph{the overlapping vertices}. We check whether the information in the overlapping vertices is consistent\footnote{For instance, the information might be inconsistent if according to the reports at $v$ the line goes through $v$ and proceeds to $w$ but according to the reports of $w$ the line does not go through $w$.} and sufficient to define the potential at all vertices of the cube $\times_i [\lfloor \oa_i \rfloor, \lceil \oa_i \rceil]$ as it was defined in Section \ref{sec:hpv}. If it is sufficient, we define $\ophi_{r^A,r^B}(\oa)$ as the multilinear extension of $\varphi$ with respect to these vertices (as we did in Section \ref{sec:rep}). Otherwise, if the information is insufficient, we define $\ophi_{r^A,r^B}(\oa)=0$. Similarly for Bob.
We notice that unlike $\varphi(a)$, that depends on the private information of Alice and Bob, the report-based potential $\varphi_{r^A,r^B}$ depends on the \emph{reports} of the players but not on their private information.

\subsubsection{Actions}
As was mentioned above, in the constructed game Alice's team (collectively) chooses a point $\oa$ and (collectively) chooses two reports $r_1^A,r_2^A$. Formally, the collective choice of a point $\oa$ is done precisely as in the replication construction (Section \ref{sec:rep}). We recall that every coordinate $\oa_i$ is an average of $m=\Theta(n)$ binary actions $a_{i,j}\in \{0,\gs\}$ of players $\{(i,j):j\in [m]\}$. The collective choice of a report $r_i^A\in \{0,1\}^{9(n+1)}$ is done by a teams of $(n+1)\cdot 3 \cdot 3$ players each with binary action $\{0,1\}$. We have $(n+1)$ neighbouring vertices, each has 3 $tsp$ bits, each one of these bits has a private information of an array of size 3.

Similarly, we define the binary actions of Bob's team. The only distinction is that the private information of Bob is indices $[3]$ which are encoded by binary strings (rather than arrays of size 3). Those the collective report of $r_i^B\in \{0,1\}^{6(n+1)}$ for $i=1,2$ is done by $6(n+1)$ players with binary actions.

\subsubsection{Utilities}
The idea is to use prioritized incentives. Roughly speaking, the incentives are \emph{prioritized} if a small improvement in a higher level priority compensates against all the possible losses of all the lower level priorities. The priorities are as follows.
\begin{enumerate}
\item At the highest priority, players are incentivized  to play close-by points $\oa$ and $\ob$ using a simple imitation utility.
\item At the high priority, a team is incentivized to report truthfully her relevant report.
\item At the medial priority, a team is incentivized to choose a point with high potential.
\item At the low priority, a team is incentivized to report truthfully the irrelevant report. 
\item Although we do not actually define utilities with respect to the following term, it is useful to mention it. A utility that is obtained in the rare event where $\oa$ is far from its expectation has the lowest weight in our arguments. Namely, even though there is some probability that the realization of $\oa$ will be far from its expectation since $m=\Theta(n)$ is large enough, this event has much smaller probability even with respect to the low priority term. 
\end{enumerate}

Formally, Alice's team has identical utility at the action profile $(\oa,r_1^A,r_2^A,\ob,r_1^B,r_2^B)=:(ab)$ that is the sum of the following terms.

\begin{enumerate}
\item Imitation loss. The imitation loss is defined by $u_A^{im}(ab)=-\sum_{i=1}^{n+1} \1_{|\oa_i-\ob_i|>1} 4^n |\oa_i-\ob_i|$. Note that for close points (i.e., $||\oa-\ob||_\infty \leq 1$) the imitation term is identically 0. We also note that the imitation term is identical for both players.

\item Cost for untruthful relevant reporting. Since we want players to be able to decrease the cost by unilateral deviation we set the cost with respect to the Hamming distance (or equivalently the $\calL_1$ distance). We define $u_A^{rr}(ab)=-3^n \1_{\oa \in R_1} d_1(r_1^A, \hat{r}_1^A)-3^n \1_{\oa \in R_2} d_1(r_2^A, \hat{r}_2^A)$, when we recall that  $R_i$ is the region where the report $r_i$ is defined to be relevant, $\hat{r}_i^A$ denotes the true private information of Alice at the relevant 5 vertices, and $d_1$ is in this case simply is the Hamming distance. The term $3^n$ is set to serve our purposes in the priority hierarchy.

Note that the cost for untruthful reporting \emph{does not} depend on Bob's actions.

\item Potential cost. The potential cost is defined by $u_A^{po}(ab)=-\ophi_{r_A,r_B}(\oa)-\ophi_{r_A,r_B}(\ob)$, where $\varphi_{r_A,r_B}$ is the report-based potential defined above. Note that unlike the replication game (Section \ref{sec:rep}) here we include the potential value of \emph{both} players in Alice's utility. This is needed to ensure that the game will be a potential game. Unlike $\varphi$ in the query model $\varphi_{r_A,r_B}$ depends on Alice's (Bob's) action because it depends on the report.

\item Cost for untruthful irrelevant reporting.  
We define $u_A^{ir}(ab)=-0.5^n \1_{\oa \notin R_1} d_1(r_1^A, \hat{r}_1^A)-0.5^n \1_{\oa \notin R_2} d_1(r_2^A, \hat{r}_2^A)$. The term $0.5^n$ is set to serve our purposes in the priority hierarchy.
\end{enumerate}
Alice's utility is given by the sum $u_A(ab)=u_A^{im}(ab)+u_A^{rr}(ab)+u_A^{po}(ab)+u_A^{ir}(ab)$.

Bob's utility is defined similarly
\begin{align*}
u_B^{im}(ab) &=-\sum_{i=1}^{n+1} \1_{|\oa_i-\ob_i|>1} 4^n |\oa_i-\ob_i|, \\
u_B^{rr}(ab) &= -3^n \1_{\ob \in R_1} d_1(r_1^B, \hat{r}_1^B)-3^n \1_{\ob \in R_2} d_1(r_2^B, \hat{r}_2^B), \\
u_B^{po}(ab) &= -2\varphi_{r_A,r_B}(\oa)-2\varphi_{r_A,r_B}(\ob), \\
u_B^{ir}(ab) &= -0.5^n \1_{\ob \notin R_1} d_1(r_1^B, \hat{r}_1^B)-0.5^n \1_{\ob \notin R_2} d_1(r_2^B, \hat{r}_2^B), \text{ and } \\
u_B(ab) &= u_B^{im}(ab)+u_B^{rr}(ab)+u_B^{po}(ab)+u_B^{ir}(ab).
\end{align*}

The game is a potential game because if we view it is a two-player game between the teams it is given as a sum of an identical interest term $u^{im}+u^{po}$ and opponent-independent term $u^{rr}+u^{ir}$; See Fact~\ref{rem:pot}. It is also easy to verify that Alice's utility does not depend on Bob's private information and vice-versa, which is needed for the communication reduction.

In the future analysis of congestion games the above property of the reduction will be useful. Hence, we provide a terminology for this property and emphasize it in a remark.
\begin{definition}\label{def:clean2}
A pair of a two-player potential game and the private information of Alice and Bob in it is called \emph{structured} if the utilities of the players can be written as $u_A(a,b)=v_C(a,b)+v_A(a)$ and $u_B(a,b)=v_C(a,b)+v_B(b)$. Moreover, we require that Alice will know $v_C$ and $v_A$ and Bob will know $v_C$ and $v_B$. 
\end{definition}

\begin{remark}\label{rem:c2}
Note that the constructed pair of a potential games with the information of Alice and Bob are structured.
\end{remark}

\subsubsection{Equilibrium analysis}

Let $(\alpha,\beta)$ be a mixed Nash equilibrium of the game, where each player randomizes between his two actions. We denote by $\oal\in [0,\gs]^{n+1}$ the expectation of Alice's team choice of the point $\oa$. Similarly, $\obe\in [0,\gs]^{n+1}$ is the expectation of Bob's team choice of the point $\ob$.

We first observe that the realized points $\oa$ and $\ob$ are $0.2$-close to $\oal$ and $\obe$ with very high probability of $1-0.1^{n}$. Namely $\Prob_{\oa \sim \alpha} [|\oa-\oal|_{\infty}>0.2]\leq 0.1^n$. This simply follows from Hoeffding inequality\footnote{Note that here we use the independent, but not identically distributed version of Hoeffding inequality.} and the choice of sufficiently large $m=\Theta(n)$.

Henceforth, we present gains by divination of at least $0.5^n$, whereas for the rare event of $|\oa-\oal|_{\infty}>0.2$ the loss in these deviations might be at most $O(n)4^{n}$. Thus we can ignore this rare event because it yields a change of at most $O(n)4^n 0.1^n << 0.5^n$ in the utility.
More formally, the \emph{essential support} of $\alpha$ is the set of realizations such that $|\oa-\oal|<0.2$. In the proof we will consider the essential supports of $\alpha,\beta$ and by the above arguments, it implies the same observations for the strategies $\alpha,\beta$.

We first show that in an equilibrium reports are truthful except for some specific region of $\oal$.

Whenever $\oal$ is close to the reference boundary of a team, it is not clear which report should the reporting players choose; Should it be the report of the vertex $v$ or that of the vertex $v+\gs e_i$? Indeed the mixed strategy of the players that choose $a$ could induce an indifference between the two reports. The following lemma states that in points that are far from the reference boundary this undesirable phenomenon does not happen, and players report truthfully.

\begin{lemma}\label{lem:truth}
Let $(\alpha,\beta)$ be a (mixed) Nash equilibrium such that $\oal$ is $0.2$-far from a reference boundary of team $r_i$ for $i=1,2$, and $\oal$ belongs to the region whose reference vertex is $v$. Then, the report of the team $r^A_i$ is the truthful report of the local information at the vertex $v$.
\end{lemma}
Similar lemma holds for Bob's two reporting teams.
\begin{proof}
In the essential support the realization of $\oa$ is such that the local information of team $i$ is of the vertex $v$.

Consider a reporter in the team $r_i^A$.
In case his report turns out to be relevant (this depends on the realization of $\oa$), his gain from truth report is $3^{n} - \Theta(1) 2^n$ where the second term counts the possible potential gains that can be obtained by false report.
If his report turns out to be irrelevant, his gain from truth report is exactly $0.5^n$ because this report does not affect other terms of the utility.
\end{proof}

As a second step, we bound the distance of $\oal$ and $\obe$ in every equilibrium, which follows from the fact that the imitation term has the highest priority.

\begin{lemma}\label{lem:abd}
In every (mixed) Nash equilibrium $(\alpha,\beta)$ we have $|\oal-\obe|_{\infty} \leq 2$.
\end{lemma}
The proof is simpler than the analogous Lemma \ref{lem:cube-bound} because we have included the replication ingredient in the game and hence we can use the essential support analysis.

\begin{proof}
Assume by way of contradiction that $\oal_i >\obe_i+2$. 
For all pair of points in the essential support $(\oa,\ob)$ we have $\oa \geq \ob + 1.6$. We argue that all players in the team of $a_i$ (i.e., the one that is responsible for the choice of $a_i$) it is profitable for them to deviate from $\gs$ to $0$.
By playing $0$ instead of $\gs$ a single player moves $\oal_i$ $\frac{\gs}{m}$-closer to $b_i$ and hence gain at least $\frac{\gs}{m}4^n$ in the imitation term. The other imitation terms (for coordinates other than $i$) remain unchanged. The loss in all other terms is bounded by $\Theta(1)3^n$. Therefore such a deviation is profitable. 
\end{proof}

From Lemma~\ref{lem:truth} and Lemma~\ref{lem:abd} we deduce that reports are truthful and are referred to the same rejoin of the line (up to adjacent vertices in the line), hence the potential $\ophi_{r^A,r^B}$ is computed correctly. Formally, we have the following lemma.

\begin{lemma}\label{lem:pot}
In every (mixed) equilibrium we have $\ophi_{r^A,r^B}(\oa)=\ophi(\oa)$ for every realization in the essential support.
\end{lemma}

\begin{proof}
The key observation is that every point $x$ that is $0.2$-close to the reference boundary of team $r_i$ the report of team $r_i$ is \emph{irrelevant}. This observation follows from the definition of the reference boundaries and the relevance boundaries which are far from each other.

Formally, if $\oal$ belongs to the region whose reference vertex is $v$ and $\oal$ is $0.2$-far from both reference boundaries, both reports $r_1^A,r_2^A$ are truthful by Lemma \ref{lem:truth}, in particular the relevant one.
If $\oal$ is $0.2$-close to the reference boundary of $r_i$, then the report $r^A_{3-i}$ is truthful by Lemma \ref{lem:truth}, and $r^A_{3-i}$ is also the relevant report.
If $\oal$ is $12.5$-far from all edges, then the report does not play any role in the definition of the potential.
Summarizing, these three cases show that the \emph{relevant} report of Alice's team is truthful for realization $\oa$ in the essential support.

Similarly, the \emph{relevant} report of Bob's team is truthful for realization $\oa$ in the essential support.

By Lemma \ref{lem:abd} $|\oal-\obe|_\infty \leq 2$. This implies that the relevant vertices of $\oal$ ($v$) and $\obe$ ($w$) are adjacent. Therefore, $v$ and $w$ are overlapping vertices with truthful reports. The local information ($tsp$) at the vertex $v$ is sufficient to define the local potential in all integer points of the cube $[[\oal]-1,[\oal]+1]$ and in particular this includes all integer points that are used in the definition of $\ophi(\oa)$ for an $\oa$ in the essential support.
\end{proof}

Lemma \ref{lem:abd} also implies that the essential support of an equilibrium is contained in a small cube (see Section \ref{sec:domdir}), which by Proposition \ref{pro:dom-dir} has a dominating direction $\pm e_i$. Similarly, to our previous arguments (Propositions \ref{theo:binary} and \ref{theo:const-actions}) we argue that players in one of the coordinate-$i$ teams ($a_i$ or $b_i$) have an incentive to deviate toward the dominating direction. The arguments are slightly more complicated because such a change has an effect on other terms of the utility (which did not appear in Propositions \ref{theo:binary} and \ref{theo:const-actions}): truthfulness of the reports. The following lemma applies such an analysis and completes the reduction.

\begin{lemma}\label{lem:final}
Let $(\alpha,\beta)$ be a (mixed) Nash equilibrium, then either $\oal$ or $\obe$ are $0.2$-close to the end-of-line vertex $l'(T)$.
\end{lemma}

\begin{proof}
Assume that $e_i$ is the dominating direction in the union of the essential supports $B(\oal,0.2)\cup B(\obe,0.2)$. We argue that each player in the team $a_i$ strictly prefer to deviate from $0$ to $\gs$.

We consider two cases. First we consider the case where $|\oal_i-\obe_i|_\infty \leq 0.5$. In this case, both teams prefer to deviate toward the dominating direction. We demonstrate it for Alice's team.

Since the potential term is equal the actual potential (Lemma \ref{lem:pot}) and the direction $e_i$ is a dominating direction, the team gains at least $\frac{\gs}{m}$ in the potential term $u^{po}_A$. This follows from the fact that the multilinear extension has a derivative of at least 1 in the dominating direction, and a unilateral deviation of a single player in the team moves the realized point $\oa$ of the team a distance of $\frac{\gs}{m}$ in the $e_i$ direction (see formal arguments in the proof of Proposition \ref{theo:binary}).

The imitation term is identically 0. Also, the irrelevant reports are negligible.\footnote{Note that indeed such deviation may cause untruthful \emph{irrelevant} reports close to the vertex $(\0_n,\gs)$ when we enter from the region of no-reports (points that are sufficiently far from the line) to points that are $13$-close to the line.}
It only remains to check that such a change does not cause a loss in the team of relevant reporters. 
If $\oal$ is 0.2-close to the reference boundary of team $i$, then the reports of team $i$ are irrelevant.
If $\oal$ is 0.2-close to the relevance boundary, then by Lemma \ref{lem:truth} in these region \emph{both} reports are truthful. Hence, changing the relevance of the reports does not cause any loss.
This completes the arguments for the first case because the total gain from such deviation is at least $\frac{\gs}{m}-\Theta(0.5^n)$.

Now we consider the second case where $|\oal_i-\obe_i|>0.5$, and we assume without loss of generality that $\oal_i<\obe_i-0.5$. 
Here we argue that Alice's team will gain by divination from $0$ to $\gs$. The only difference from the previous case is that we cannot argue that the imitation term is negligible. However, for every pair of points in the essential support we know that a deviation from $a_i$ to $a_i+\frac{29}{m}$ decreases the distance $|a_i-b_i|$. Therefore such a divination \emph{improves} the imitation loss. In addition to this argument, all the remaining arguments about the negligibility of the other changes in the utility remains unchanged.
\end{proof}

Lemma \ref{lem:final} completes the proof of Theorem \ref{theo:np} because the constructed game has $O(n^2)$ players, and end-of-line can be computed (without communication) from any point in the support of an equilibrium. 

\subsection{Two-player hardness}

We use the high-degree imitation game (Section \ref{sec:hd}) combined with the idea of reporting the local information as in the previous Section \ref{sec:n-pl}. Unlike the $n$-player case, here we do \emph{not} have the obstacle of different players who are responsible  for choosing the point and choosing the report. This obstacle required delicate treatment in the $n$-player case. In the two-player case, we can simply define that player reports local information of the point chosen by him. A player can safely deviate to a point where the local information report varies, and report the corresponding information. We define the game formally, and briefly overview the proof that is very similar to the proof of Theorem \ref{theo:const-actions} with one additional (trivial) argument that players report truthfully their local information.

\paragraph{Local information} As was mentioned above the notion of truthful report is defined much simpler in the two-player case. For a point $a\in [29]_0$ that is $12$-close to a vertex $v\in \{0,29\}^n$ we define $T_A(a)\in \{0,1\}^{9(n+1)}$ to be the array of Alice's local information with respect to the vertex $v$. We recall that Alice's local information of a vertex consists of three $tsp$ ternary arrays for the vertex and its $n$ neighbours. For a point $a\in [29]_0$ that is $12$-far from all vertices $v\in \{0,29\}^n$, but is $12$-close to an edge $(v,v+29e_i)$ again we define $T_A(a)$ with respect to the vertex $v$.

Bob's local information is defined similarly, with one distinction. The local informant with respect to a vertex $v$ is an element of $T_B(b)\in [3]^{3(n+1)}$ because Bob's private information are indices (rather than arrays).

\paragraph{Actions}
Alice chooses a point $a\in [29]^n_0$ in the grid, and a report of local information $r^A\in \{0,1\}^{9(n+1)}$. 
Similarly Bob chooses a point $b\in [29]^n_0$ in the grid, and a report of local information $r^B\in [3]^{3(n+1)}$.

\paragraph{Utilities}
The utility of Alice is defined by three priorities
\begin{enumerate}
\item Alice's high priority term is the high-degree imitation $u_A^{im}(a,b)=-\sum_{i=1}^{n+1} (a_i-b_i)^{2n}$.

\item Alice's medial priority term is the truthful reporting $u_A^{r}(a,r^A)=2.5^n \1_{r^A=T_A(a)}.$

\item Alice's low priority term is the report-based potential $u_A^{po}(a,r^A,b,r^B)= 2\phi_{r^A,r^B}(a)+2\phi_{r^A,r^B}(b)$. The report based potential $\phi_{r^A,r^B}$ is defined as in Section \ref{sec:n-pl}. Namely, if reports overlap, are consistent, and the value of the potential can be deduced from the reports we define $\phi_{r^A,r^B}$ to be this value. Otherwise, we set $\phi_{r^A,r^B}=0$.

\end{enumerate}

Finally, we set $u_A(a,r^A,b,r^B)=u_A^{im}+u_A^{r}+u_A^{po}$.
Bob's utility is defined similarly.

The game is a potential game because it is given as a sum of an identical interest term $u^{im}+u^{po}$ and opponent-independent term $u^{r}$; see Fact~\ref{rem:pot}.

In the future analysis of congestion games the above property of the reduction will be useful. Hence, we provide a terminology for this property and emphasize it in a remark.
\begin{definition}\label{def:cleann}
A pair of a $2n$-player potential game with Alice's and Bob's private information in it is called \emph{structured} if the utilities of all Alice's players are identical. The utilities of all Bob's players are identical.
These identical utilities can be written as $u_A(a,b)=v_C(a,b)+v_A(a)$ and $u_B(a,b)=v_C(a,b)+v_B(b)$. Moreover, we require that Alice knows $v_A$ and $v_C$ and Bob knows $v_B$ and $v_C$.
\end{definition}

\begin{remark}\label{rem:cn}
Note that the constructed potential games with Alice's and Bob's private information in the reduction are structured.
\end{remark}

\paragraph{Equilibrium analysis} Let $(\alpha,\beta)$ be a mixed Nash equilibrium.

First, we apply the arguments of Lemma \ref{lem:hd} to deduce\footnote{The only change with respect to the proof of Lemma \ref{lem:hd} is the fact that now we have a reporting term $u^r$. However, since $u^r\leq 2.5^n$ whereas the gains in the imitation term are of order $3^n$ precisely the same arguments can be applied.} that in every equilibrium
$\supp(\alpha)\cup \supp(\beta)$ is contained in a small cube (of edge-length 4).

Second, we note that reports do not affect the imitation term. By truthful report, a player gains $2.5^n$, whereas his loss in the potential term is at most $\Theta(2^n)$. Hence, in a Nash equilibrium players report truthfully.

Third, since the support of the players is contained in a small cube we can deduce that the reported information will overlap, and the potential could be computed from the reports. Since the reports are truthful we also know that $\phi_{r^A,r^B}(a)= \phi(a)$ and similarly for the point $b$.

Finally, we apply the arguments of Proposition \ref{theo:const-actions} to deduce that the only mixed Nash equilibrium is pure, and in this equilibrium, both players choose the end-of-line point.

\section{Congestion games}\label{sec:congestion}
A congestion game is defined with respect to $n$ players $i=1,2,...,n$ and $N$ facilities $j=1,2,...,N$. Each facility has a (common to all users) cost function $c_j:[n] \rightarrow \mathbb{R}$ that indicates the utility of a player that chooses this facility. An action of player $i$ is a collection of feasible subsets $f_i\subset [N]$. The collection of all his feasible subsets is denoted by $F_i=\{f_i^1,...,f_i^{K_i}\}$. Any profile of actions $f=(f_i)_{i\in [n]}$ where $f_i\in F_i$ defines a vector of congestions on the facilities $g_j(f):=|\{i:j\in f_i\}|$. The utility of player $i$ in the congestion game is given by $U_i(f)=\sum_{j\in f_i} c_j(g_j(f))$.

When congestion games are discussed in the context of communication complexity we have to specify what is the private information of each player. We focus on the following variant. The collections $(F_1,...,F_n)$ are common knowledge. Namely, players are aware of the possible actions of their opponents. Each player $i$ knows the cost functions of all possible facilities that he can potentially use, i.e., he knows $\{c_j : j\in \cup_{k\in [K_i]} f^k_i \}$. The only information that is missing for agent $i$ is the cost functions of the facilities that he cannot use. 
As we show in Corollaries \ref{cor:cong2} and \ref{cor:congn} this missing information is crucial for players' ability to find equilibrium with small communication.

Note that with this distribution of information players know slightly more than in the standard \emph{uncoupled} distribution of information. 
Uncoupledness requires that player $i$ will know his utility function $U_i=\sum c_j(\cdot)$. We allow the player to learn each cost function $c_j(\cdot)$ separately.
In other words, our communication complexity lower bounds imply that uncoupled dynamics cannot efficiently converge to equilibrium, even if they are enhanced with the additional knowledge of costs of each facility.

The equivalence of potential games and congestion games was already shown in Monderer and Shapley's seminal paper~\cite{MS}. 
We cannot apply this reduction black-box to our hardness of potential games because defining the congestion costs are defined using the potential, but in the communication complexity it is crucial that players cannot directly evaluate the potential function. (If players had offline access to the potential function, they could find its maximizer --which is a pure equilibrium-- with zero communication!)
Nevertheless we can use the particular structure of our hard potential games (see Definitions \ref{def:clean2} and \ref{def:cleann}) to extend the hardness to congestion games.

\begin{proposition}\label{pro:c2}
Every 2-player $N$-action potential game with information distribution that is structured can be presented as a congestion game with $N^2+2N$ facilities where Alice and Bob know the  cost functions of their possible facilities.
\end{proposition}

\begin{proof}
We denote by $A$ and $B$ (with $|A|=|B|=N$) the action sets of Alice and Bob correspondingly.
We denote the utilities of the structured potential game by 
\begin{align*}
u_A(a,b)=v_C(a,b)+v_A(a) \\
u_B(a,b)=v_C(a,b)+v_B(b)
\end{align*}
we define the set of facilities to be $A\cup B \cup (A\times B)$. The facilities $A$ are private for Alice, the facilities $B$ are private for Bob, and the facilities $A\times B$ are common to Alice and Bob. The feasible subsets of facilities for Alice are $F_a:=\{a\}\cup \{(a,b):b\in B\}$ for all $a\in A$. The feasible subsets of facilities for Bob are $F_b:=\{b\}\cup \{(a,b):a\in A\}$ for all $b\in B$. Note that indeed Alice never chooses facilities in $B$ and conversely Bob never chooses facilities in $A$.

The cost functions of a facility $a\in A$ is given by $c_a(1)=v_A(a)$ (we should not specify the cost of two players because it never happens). The cost functions of a facility $b\in B$ is given by $c_b(1)=v_B(b)$. The cost functions of a facility $(a,b)\in A\times B$ is given by $c_{(a,b)}(1)=0$ and $c_{(a,b)}(2)=v_C(a,b)$. It is easy to verify that indeed the constructed congestion game has the utility profile $(u_A(a,b),u_B(a,b))$ at the action profile $(F_a,F_b)$ as needed.
\end{proof}

\begin{corollary}\label{cor:cong2}
The communication complexity of finding a mixed Nash equilibrium in two-player $N$-facility congestion game is at least $N^c$ for some fixed constant $c>0$.
\end{corollary}

\begin{proposition}
Every $2n$-player binary-action potential game with information distribution that is structured can be presented as a congestion game with $2^{2n}+2^{n+1}$ facilities where Alice and Bob know the cost functions of their possible facilities.
\end{proposition}

\begin{proof}
We denote by $A=\{0,1\}^n$ and $B=\{0,1\}^n$ the action sets of Alice's group and Bob's group correspondingly.
We denote the common utilities of the two groups in the structured potential game by 
\begin{align*}
u_A(a,b)=v_C(a,b)+v_A(a) \\
u_B(a,b)=v_C(a,b)+v_B(b)
\end{align*}
we define the set of facilities to be $A\cup B \cup (A\times B)$. The facilities $A$ are private for all players in Alice's group, the facilities $B$ are private for all players in Bob's group, and the facilities $A\times B$ are common. The two feasible subsets of facilities for player $i$ in Alice's group are $F^{A,i}_0:=\{a\in A:a_i=0\}\cup \{(a,b)\in A\times B:a_i=0\}$ and $F^{A,i}_1:=\{a\in A:a_i=1\}\cup \{(a,b)\in A\times B:a_i=1\}$. Similarly for Bob. Note that indeed players in Alice's group never choose facilities in $B$ and conversely players in Bob's group never choose facilities in $A$.

The cost functions of a facility $a\in A$ is given by $c_a(k)=0$ if $k<n$ and $c_a(k)=v_A(a)$ if $k=n$ (we should not specify the cost of more than $n$ players because it never happens). The cost functions of a facility $b\in B$ is given by $c_b(k)=0$ if $k<n$ and $c_b(k)=v_B(b)$ if $k=n$. The cost functions of a facility $(a,b)\in A\times B$ is given by $c_{(a,b)}(k)=0$ if $k<2n$ and $c_{(a,b)}(k)=v_C(a,b)$ if $k=2n$. It is easy to verify that indeed the constructed congestion game has the common utility $u_A(a,b)$ for all players in Alice's group if players play the action profile $F_{\cdot}$ that corresponds to $(a,b)$. Similarly for Bob's group.
\end{proof}

\begin{corollary}\label{cor:congn}
The communication complexity of finding a mixed Nash equilibrium in $n$-player $2^n$-facility binary-action congestion game is at least $2^{c\sqrt{n}}$ for some fixed constant $c>0$.
\end{corollary}

We note that our negative result cannot be further improved to games with few ($\poly(n)$) facilities. Those games have succinct representation and hence can be solved with polynomial communication.

\bibliography{ref}

\appendix
\section{Proof of Lemma \ref{lem:hd}}\label{ap:hdlemma}

Let $y_k$ be the probability that $y$ assigns to the integer number $k\in [29]_0$. We denote $z_k=\sqrt[2n]{y_k}$.
For ``typical distributions" we have $z_k\approx 1$ for the actions in the support. However, if $y_k$ is exponentially small, we may have $z_k$ that is bounded away from 1.

We define $\gsp$ convex functions $f_k:[0,\gs]\ra \R$ by $f_k(x)=z_k|x-k|$. We consider the convex function $f:[0,\gs]\ra \R$ defined by $f(x)=\max_{k=0,...,\gs} f_k(x)$. 

We argue that $f$ is not too flat around its minimum, and gets high enough values for points that are not close to the minimum. Specifically, we denote by $x^*\in [0,\gs]$ the (real) minimum of $f$ over $[0,\gs]$. We denote by $m_1,m_2,m_3\in [\gs]_0$, $m_i\neq m_j$, $f(m_1)\leq f(m_2) \leq f(m_3)$ the three minimal integer values of $f$. 
Note that since $f$ is convex the set $\{m_1,m_2,m_3\}$ forms a segment of length $2$.
Let $a\in [\gs]_0\setminus \{m_1,m_2,m_3\}$ be any other integer. 
We argue that 
\begin{align}\label{eq:m4}
f(a)\geq 1.8 \text{ and } f(a)\geq f(m_1)+\frac{1}{30}.
\end{align}

To see that $f(a)\geq 1.8$, we observe that at least one of the functions $f_k$ is very close to the function $|x-k|$. There exists $m\in [\gs]_0$ such that $y_m\geq \frac{1}{\gsp}$, and therefore $z_m\geq \sqrt[2n]{\frac{1}{\gsp}}\geq 0.9$. The fourth minimal value at integers points of $f_k(x)\geq 0.9|x-k|$ is at least $2\cdot 0.9=1.8$ because the froth most close integer to $k$ is within distance $2$ from $k$. Obviously, since $f(x)\geq f_k(x)$ we have $f(a)\geq 1.8$.

To see that $f(a)\geq f(m_1)+\frac{1}{30}$, we consider two cases.

If $f(x^*)\leq 1$, we observe that $|f'(x)|\leq 1$ for all $x\in [0,29]$ where the derivative $f'(x)$ is defined. Therefore $f(x)-f(x^*)\leq |x-x^*|$, and in particular $f(m_1)\leq f([x^*]) \leq f(x^*)+|[x]-x^*| \leq 1+0.5$. Hence, $f(a)-f(m_1)\geq 1.8-1.5 \geq 0.3 \geq \frac{1}{30}$.

Now we prove that $f(a)-f(m_1)$ in case where $f(x^*)>1$. For clarity of exposition we denote $a=m_4$ (although it is not necessarily the forth minimal integer value).
At least two integers among $\{m_i\}_{i\in [4]}$ are located on the same side of $x^*$. Without loss of generality we assume that $x^*\leq m_i < m_j$.
We denote by $f'_R(x^*)$ the derivative of $f$ from the right side. Since $x^*$ is minimum we know that $f'_R(x^*)>0$. Moreover, we observe that since $f(x^*)>1$ we necessarily have $f'_R(x^*)\geq \frac{1}{30}$. This follows from the fact that every function $f_k$ that achieves a value $f_k(x)\geq 1$ in the range $[0,\gs]$ must have a derivative of at least $|f'_k(x)|\geq \frac{1}{30}$, otherwise it could not change its value from $0$ to $1$ in this segment. 
We complete the arguments by
\begin{align*}
f(a)-f(m_1)\geq f(m_j)-f(m_i) \geq f'_R(x^*)(m_j-m_i) \geq \frac{1}{30} 
\end{align*}
where the first inequality follows from $f(a)\geq f(m_j)$ and $f(m_1)\leq f(m_i)$. The second inequality follows from the concavity of $f$.
This accomplishes the proof of Equation \eqref{eq:m4}.

The expected loss at $a$ can be bounded from below by

\begin{align*}
\sum_{k=0}^{29} y_i (a-k)^{2n} &\geq \max_{k\in [29]_0} y_i (a-k)^{2n}=\max_{k\in [29]_0} [(f_k(a))^{2n}]
=[\max_{k\in [29]_0} (f_k(a))]^{2n}=(f(a))^{2n}
\end{align*}

The expected loss at $m_1$ can be bounded from above by

\begin{align*}
\sum_{k=0}^{29} y_i (m_1-k)^{2n} &\leq 30 \max_{k\in [29]_0} y_i (m_1-k)^{2n}=(f(m_1))^{2n}
\end{align*}

Therefore the difference in the losses between the sub-optimal choice of $a$ and the optimal choice of $m_1$ is at least
\begin{align}\label{eq:po}
f(a)^{2n}-30 f(m_1)^{2n} \geq 1.8^{2n}-30(1.8-\frac{1}{30})^{2n}>3.24^n-30 \cdot 3.13^n >3^n
\end{align}
Where the first inequality follows from Equation \eqref{eq:m4}.

\end{document}